\documentclass{llncs}
\usepackage{xr-hyper}
\usepackage[margin=2.5cm]{geometry}
\usepackage[utf8]{inputenc}
\usepackage[T1]{fontenc}
\usepackage{balance}
\usepackage{booktabs}     
\usepackage{amsfonts}     
\usepackage{nicefrac}      
\usepackage{graphicx}
\usepackage{breqn}
\usepackage{todonotes}
\usepackage{xcolor}
\usepackage{enumitem}
\usepackage{rotating}
\usepackage{array}
\usepackage{mathtools}

\DeclarePairedDelimiter\floor{\lfloor}{\rfloor}
\usepackage{wrapfig}
\usepackage{url,lineno,subcaption}
\usepackage{blindtext}
\usepackage[ruled,linesnumbered]{algorithm2e}
\usepackage[hidelinks]{hyperref}
\usepackage{bm}
\usepackage{ulem}
\newcommand{\ind}{\perp\!\!\!\!\perp}
\newcommand{\e}{\mathbb{E}}

\usepackage{amsmath,amssymb}

\newcommand\numberthis{\addtocounter{equation}{1}\tag{\theequation}}
\makeatletter
 \def\@textbottom{\vskip \z@ \@plus 1pt}
 \let\@texttop\relax
\makeatother

\usepackage{url}
\usepackage[labelfont=bf]{caption}

\begin{document}

\title{CITS: Nonparametric Statistical Causal Modeling for High‑Resolution Time Series}

\author{Rahul Biswas\inst{1,5} \and SuryaNarayana Sripada\inst{2} \and Somabha Mukherjee\inst{3} \and Reza Abbasi-Asl\inst{1,4,5,}\thanks{Corresponding author: Reza.AbbasiAsl@ucsf.edu}}

\institute{
\hspace{-.4em}Department of Neurology, University of California, San Francisco, CA 94158, USA.
\and
\hspace{-.4em}Center for Research on Science and Consciousness, Redmond, WA 98052, USA.
\and
\hspace{-.4em}Department of Statistics and Data Science, National University of Singapore, 117546, Singapore.
\and
\hspace{-.4em}Department of Bioengineering and Therapeutic Sciences, University of California, San Francisco, CA 94158, USA.
\and
\hspace{-.4em}UCSF Weill Institute for Neurosciences, San Francisco, CA 94158, USA.
}

%\author{Rahul Biswas, SuryaNarayana Sripada, Somabha Mukherjee}

% \author[1,*]{Rahul Biswas}
% \author[2]{SuryaNarayana Sripada}
% \author[3]{Somabha Mukherjee}
% \affil[1]{Department of Electrical and Computer Engineering, University of Washington, Seattle, 98195, USA.}
% \affil[2]{Center for Research on Science and Consciousness, Redmond, 98052, USA.}
% \affil[3]{Department of Statistics and Data Science, National University of Singapore, 117546, Singapore.}
% \affil[*]{rbiswas1@uw.edu}

\maketitle

\begin{abstract} Identifying causal interactions in complex dynamical systems is a fundamental challenge across the computational sciences. Existing functional connectivity methods capture correlations but not causation. While addressing directionality, popular causal inference tools such as Granger causality and the Peter–Clark algorithm rely on restrictive assumptions that limit their applicability to high-resolution time-series data, such as the large-scale recordings now standard in neuroscience. Here, we introduce CITS (Causal Inference in Time Series), a nonparametric framework for inferring statistically causal structure from multivariate time series. CITS models dynamics using a structural causal model of arbitrary Markov order and statistical tests for lagged conditional independence. We prove consistency under mild assumptions and demonstrate superior accuracy over state-of-the-art baselines across simulated linear, nonlinear, and recurrent neural network benchmarks. Applying CITS to large-scale neuronal recordings from the mouse visual cortex, thalamus, and hippocampus, we uncover stimulus-specific causal pathways and inter-regional hierarchies that align with known anatomy while revealing new functional insights. We further highlight CITS ability in accurately identifying conditional dependencies within small inferred neuronal motifs. These results establish CITS as a theoretically grounded and empirically validated method for discovering interpretable statistically causal networks in neural time series. Beyond neuroscience, the framework is broadly applicable to causal discovery in complex temporal systems across domains.

\end{abstract}
%\externaldocument*{supplementary}

\section{Introduction}
Inferring causality from time series data is a fundamental challenge across disciplines such as neuroscience \cite{reid2019advancing,biswas2022statistical1}, econometrics \cite{gokmenoglu2019time}, climatology \cite{barbero2018temperature}, and geosciences \cite{massmann2021causal}. In neuroscience, understanding how neural interactions give rise to perception, cognition, and behavior requires identifying directed influences between neurons-termed the causal functional connectome (CFC) \cite{reid2019advancing,biswas2022statistical1}. The CFC comprises a directed graph in which an edge from neuron $u$ to $v$ indicates that $u$'s activity at time $t$ has a causal effect on $v$'s activity at a later time $t'$. Reconstructing this graph offers critical insights into brain computation and may serve as a biomarker for neurological disorders such as Alzheimer’s disease \cite{nakamura2017early,biswas2023causalfcn,biswas2024causalAD}.

While identifying the true causal relationships in the brain requires extensive and expensive bench experiments, statistical modeling offers an alternative solution to profile putative relationships that are most likely to be causal. Most existing methods for time series causal inference rely on parametric models that impose structural assumptions on the data-generating process. Classical implementations of Granger Causality, for example, typically assume linear vector autoregressive (VAR) models with Gaussian noise \cite{barnett2009granger,granger2001essays}. Extensions to non-linear or additive models exist \cite{chu2008search}, but these remain constrained by fixed functional forms. Non-parametric approaches avoid such limitations by inferring causal graphs through statistical independence testing rather than explicit dynamical modeling. These include methods based on directed graphical models, such as the Time-Aware Peter-Clark (TPC) algorithm \cite{biswas2022statistical2,spirtes2000causation}, and structural causal models (SCMs) \cite{peters2013causal}. A detailed review is provided in \textbf{Supplementary~\ref{appsec:cinf_tsreview}}.

Additionally, most CFC methods-whether based on correlation, coherence, or mutual information, do not resolve the directionality of influence or the underlying mechanisms through which information propagates across brain areas \cite{smith2011network}. This limitation has spurred interest in statistically causal connectivity approaches that aim to infer the directed influence of one neural unit on another. Granger Causality (GC) has been widely used for this purpose, but its reliance on linear autoregressive modeling and assumptions of Gaussian noise limits its applicability to complex or high-frequency neural data \cite{barnett2009granger}. Moreover, classical methods based on Peter-Clark (PC) algorithm, while offering a non-parametric alternative, typically assume i.i.d. observations and are not designed to capture temporally lagged dependencies in time series \cite{kalisch2007estimating,spirtes2000causation}.

To address these challenges, the structural causal modeling (SCM) framework offers a principled solution for statistical causal inference with several key advantages: 1. It provides a graphical approach that is interpretable and facilitates the verification of causal relationships of interest; 2. It avoids reliance on strong distributional or functional assumptions often imposed in time series analysis, such as linearity or multivariate normality; 3. It is grounded in the Neyman–Rubin causal model \cite{holland1986statistics} and extends earlier work on causality using linear structural equations \cite{haavelmo1943statistical}. Despite these strengths, relatively little attention has been given to incorporating Markovian assumptions into SCM for time series data. However, the Markovian framework is widely used in practice across fields such as neuroscience and econometrics \cite{korda2016discrete,pham2017texture}. One of its key advantages is that it restricts causal influences to a finite time window preceding the current observation, which aligns well with the temporal structure of real-world systems.

To operationalize this framework, we introduce the Causal Inference in Time Series (CITS) algorithm for discovering causal structure from multivariate time series. CITS is based on non-parametric SCM for time series which is Markovian of an arbitrary but finite order \cite{peters2017elements,li2017nonparametric}. CITS estimates the unrolled directed acyclic graph (DAG) by testing conditional independencies among variables within a $2\tau$ temporal window, accommodating both lagged and concurrent effects. It supports statistical tests suited to the underlying data distribution-such as partial correlation for approximately Gaussian data, or the Hilbert-Schmidt Independence Criterion for nonlinear, non-Gaussian settings. Once the unrolled DAG is recovered, we construct a rolled causal summary graph that captures directed influences across time. Under standard assumptions-including stationarity, faithfulness, and consistency of the independence test-CITS is provably correct. Moreover, when the process satisfies a first-order Markov condition and lacks concurrent interactions - such as in high-resolution electrophysiology, CITS remains robust to latent confounding, making it particularly valuable in neural systems where some sources are unobserved.

We evaluate CITS on a suite of synthetic time series benchmarks, including linear Gaussian, nonlinear non-Gaussian, and continuous-time recurrent neural networks. In each setting, CITS consistently recovers the correct causal structure more accurately than classical methods such as Granger Causality, the PC algorithm, and its time-aware extension (TPC), particularly in the presence of nonlinear dependencies and both source and sink in causal effects. We demonstrate the method’s practical utility in neuroscience to obtain functional neural circuitry under different visual stimulus conditions by applying CITS to high resolution electrophysiology recordings from mouse visual brain. We also uncovered the neural signal interactions within small neuronal motifs in the inferred pathways. The resulting causal graphs and their region-to-region anatomical distribution reveal biologically interpretable, stimulus-specific interactions among cortical, thalamic, and hippocampal regions, highlighting the potential of CITS to uncover directed neural interactions from large-scale electrophysiological time series. %The resulting causal graphs reveal biologically interpretable, stimulus-specific interactions among cortical, thalamic, and hippocampal regions-highlighting the potential of CITS to uncover directed neural interactions from large-scale electrophysiological time series.

\section{CITS: A Statistically Causal Inference Framework for High-Resolution Time Series Data}

We introduce a novel framework for causal inference in multivariate time series, grounded in structural causal models (SCMs). Our approach captures both lagged and same-timestep (concurrent) effects by modeling causal dynamics as a finite-order Markovian process, aligning with many real-world measurements such as neural recordings. The method comprises three key steps: modeling causal relationships via an SCM, identifying them through conditional independencies in the data, and inferring them using a flexible, non-parametric algorithm (CITS) that supports both Gaussian and non-Gaussian settings. In a key special case of first-order Markov processes without concurrent effects, CITS enables confounding-robust inference, making it particularly suited to neural applications.

\begin{figure}[t!]
    \centering
    \includegraphics[width = \textwidth]{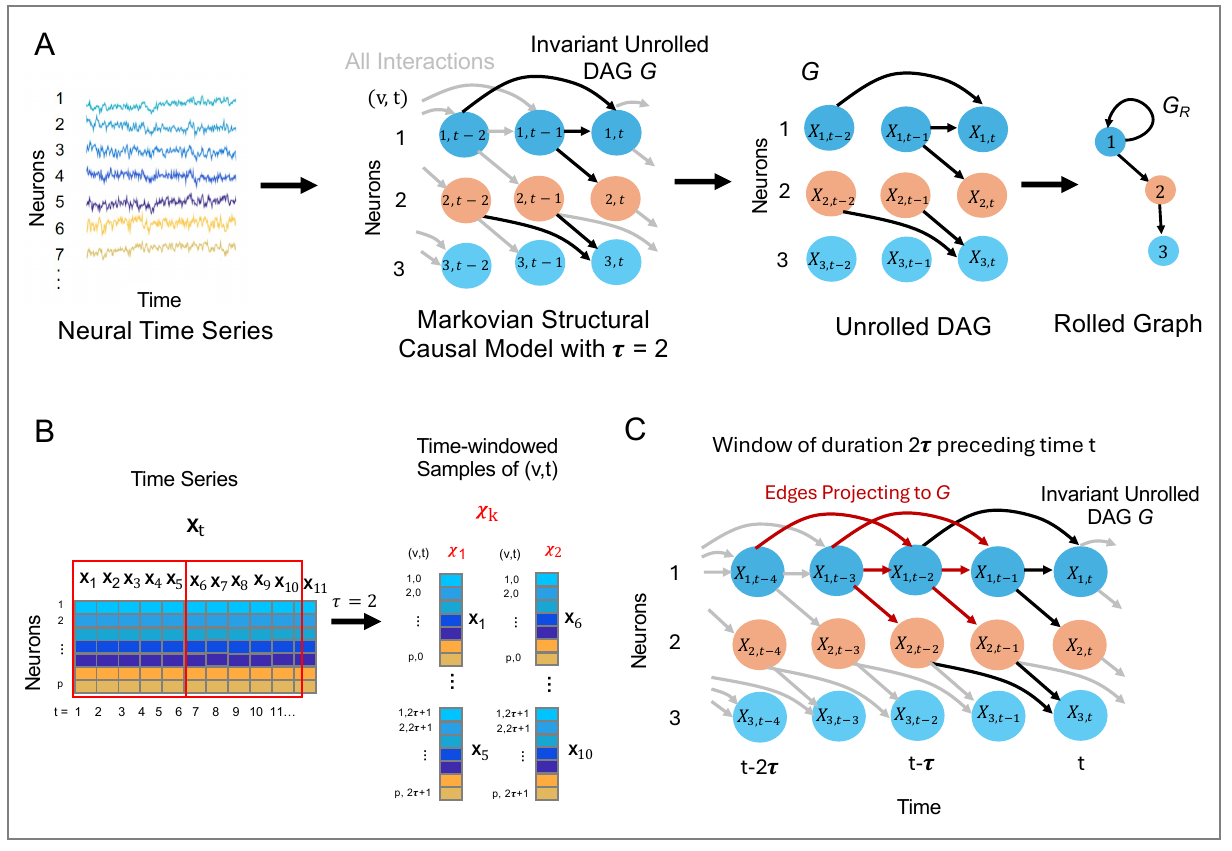}
    \caption{\textbf{Inference of the Unrolled Directed Acyclic Graph (DAG) for Neural Time Series and its Rolled Graph.} (A) Example Markovian Structural Causal Model interactions with $\tau = 2$. (B) Formation of time-windowed samples for each $(v,t)$, where $0 \leq t \leq 2\tau$. (C) The invariant unrolled DAG extends from time $t-\tau$, with edges projecting into it originating from at most time $t-2\tau$, motivating conditional dependence tests within a $2\tau$ window.}
    \label{fig:graphexample}
\end{figure}

\subsection{Markovian Structural Causal Model for Time Series}\label{sec:tsscm}

In the time series setting we consider, the data consists of a finite realization of a strictly stationary multivariate Markovian process $\{\mathbf{X}_t\}_{t\in \mathbb{Z}}$ of order $\tau$ with $p$ components, i.e., $\mathbf{X}_t = (X_{1,t},\ldots,X_{p,t})$ for every $t\in \mathbb{Z}$. The number of components $p$ is arbitrary but fixed. Serial dependence may exist both within and across components. We also assume that the stochastic process satisfies a structural causal model (SCM) that remains invariant across time. The SCM consists of a collection of assignments:
\begin{equation}\label{eqdef: process}
X_{v,t} = f_{v}(X_{\text{pa}(v,t)}, \epsilon_{v,t}), \quad v=1,\ldots,p, \ t\in \mathbb{Z},
\end{equation}
where $\text{pa}(v,t)\subseteq \{(d,s): d=1,\ldots,p;\ s=t,\ldots,t-\tau\}$, the $\epsilon_{v,t}$ are jointly independent across $v$ and $t$, and for $S\subseteq \{1,\ldots,p\}\times \mathbb{Z}$, we define $X_S := (X_{\bm{i}})_{\bm{i}\in S}$.

This formulation allows for concurrent (i.e., same-time) effects, where a variable $X_{u,t}$ may directly influence another variable $X_{v,t}$ at the same time step. In later sections, we focus on the important special case where $\text{pa}(v,t)$ excludes concurrent effects (i.e., $s < t$). This assumption, often valid in high-resolution time series, enables identifiability even under latent confounding when combined with a first-order Markov structure.

The graph $G$ associated with the SCM is constructed by creating one vertex for each $(v,t)$, where $v=1,\ldots,p$ and $t=1,\ldots,\tau+1$. A directed edge is drawn from each parent $(u,s) \in \text{pa}(v,\tau+1)$ to the node $(v,\tau+1)$, reflecting the functional dependence in \eqref{eqdef: process}. For notational simplicity, we may occasionally refer to the vertex $(v,t)$ as $X_{v,t}$. We assume that $G$ is acyclic, in which case it is a Directed Acyclic Graph (DAG). An illustrative example is shown in \textbf{Fig.~\ref{fig:graphexample}A}.

By the assumption of stationarity, the causal graph $G$ remains invariant across time: for any $t \geq \tau + 1$, the structure of directed edges from $X_{\text{pa}(v,t)}$ to $X_{v,t}$ (for all $v = 1, \ldots, p$) replicates the same DAG $G$ (see \textbf{Fig.~\ref{fig:graphexample}A}). Moreover, under the Markovian assumption of order $\tau$, the set $\text{pa}(v,t)$ includes variables indexed at time points $s = t, t - 1, \ldots, t - \tau$ and no earlier.

We further assume that the SCM satisfies the faithfulness assumption, which implies that all and only those conditional independencies present in the distribution of the process are entailed by the DAG $G$. Under this assumption, $G$ can be identified up to its Markov equivalence class from the joint distribution of the process. Notably, in the time series context, the temporal ordering disambiguates the direction of non-concurrent effects: all lagged causal relationships (from time $s < t$ to $t$) are uniquely identifiable as they respect the inherent time order. The structure among variables at the same time step is identifiable only up to a Markov equivalence class, akin to constraint-based methods like the PC algorithm, due to the lack of temporal ordering within a single time slice \cite{spirtes2000causation}.

The graph $G$, constructed over variables indexed by $(v, t)$ with edges from $X_{u, t_1}$ to $X_{v, t_2}$ representing causal influence across time, is often referred to as the \textit{unrolled DAG}. This representation explicitly captures the temporal evolution of causal dependencies over the sequence of observations. For interpretability and analysis, the unrolled DAG is often summarized into a \textit{rolled graph}, also known as a \textit{summary causal graph} \cite{peters2013causal}. The rolled graph $G_R$ is a directed graph over the $p$ components (variables), with an edge $u \to v$ if variable $X_{u, t}$ has a causal effect on $X_{v, t'}$ at some later (or same) time $t' \geq t$.

\begin{definition}\label{def:rolledgraph}
The Rolled Graph of $G$, denoted by $G_R$, is the directed graph over nodes $1,\ldots,p$ with an edge $u \to v$ if and only if there exists an edge $(u, t_1) \to (v, t_2)$ in the unrolled DAG $G$ for some $t_1 \leq t_2$.
\end{definition}

Due to the stationarity and Markovian properties of the process, it suffices to examine causal edges into a fixed time point $t$ from time points $t - \tau, \ldots, t$. Consequently, an edge $u \to v$ exists in the rolled graph $G_R$ if and only if there is a directed edge $X_{u,s} \to X_{v,t}$ in the unrolled DAG $G$ for some $s \in \{t-\tau,\ldots,t\}$. For example, one can fix $t = 2\tau + 1$ and check whether any $X_{u,s}$ with $s = \tau+1, \ldots, 2\tau$ has a directed edge to $X_{v,2\tau+1}$ in $G$, as illustrated in \textbf{Fig.~\ref{fig:graphexample}A}.

\subsection{Conditional Independence and Graph Recovery in Time Series}\label{sec:properties}
The goal of this paper is to estimate $G$ and thereafter $G_R$ in a non-parametric manner without relying on particular model specifications for the underlying time series. To motivate the goal, we consider a simple example of a stationary VAR(p)-model: $X_{v,t} = \sum_{u=1}^p \sum_{j=1}^{\tau} \phi_{uv}^{(j)} X_{u,t-j}+\epsilon_{v,t}$, where $\tau$ is the order of the Markovian process, and the noise variables $\epsilon_{v,t}$ are i.i.d. with mean zero and $\epsilon_{1,t},\ldots,\epsilon_{p,t}$ are independent of $\{X_{u,s}:u=1,\ldots,p, s<t\}$. Then, the entry of the adjacency matrix of $G$ corresponding to the edge $(u,t-j)\rightarrow (v,t)$ in $G$ is, \begin{equation}\label{eq:var_adj}\mathbf{1}(\phi_{uv}^{(j)}\neq 0).\end{equation}

In this scenario, the weights $\phi_{uv}^{(j)}$ can be estimated by a Likelihood Ratio (LR) test assuming Gaussian distributed noise terms, and plugging them into \eqref{eq:var_adj}, one can estimate the adjacency matrix of $G$ and therefore $G$. This method underpins classical Granger causality, where directed edges in $G$ are inferred via significance tests on the VAR coefficients. However, if the noise distribution is unknown, LR-based inference becomes unreliable. Furthermore, if the true underlying data generating stationary process is nonlinear, perhaps with non-additive innovation terms, it is non-trivial to extend this approach. 

Instead, we adopt a nonparametric strategy based on conditional independence (CI), inspired by constraint-based methods such as PC algorithm \cite{kalisch2007estimating}. This approach leverages temporal order, the Markov property, and stationarity to constrain the space of valid conditioning sets. The adjacency of edges in $G$ can be determined using a conditional independence oracle as follows. Under the faithfulness assumption, $X_{v,t}$ and $X_{u,s}$ (for $t-\tau \leq s\leq t$) are non-adjacent in $G$ if and only if there exists a set $S$ such that they are d-separated by $S$, which in turn is equivalent to their conditional independence given $S$ \cite{verma2022equivalence}. Moreover, $X_{v,t}$ and $X_{u,s}$ are d-separated by their respective parents if and only if they are non-adjacent \cite{verma2022equivalence}. Since the process is Markovian of order $\tau$, the parents of both nodes are restricted to the interval ${s - \tau, \ldots, t}$. Consequently, conditional independence of $X_{v,t}$ and $X_{u,s}$ given a subset of nodes in this interval implies their non-adjacency.

This illustrates that we can relate the adjacency of a pair of nodes $X_{v,t}$ and $X_{u,s}$ in $G$ to conditional independence information of the pair of nodes given a set of other nodes in the interval $\{t\wedge s-\tau, \ldots, t\vee s\}$ (see \textbf{Fig. \ref{fig:graphexample}-C}). This is formalized by the following lemma, proved \textbf{in Supplementary \ref{appsec:theo_guar}}.

\begin{lemma}\label{lemma:concept}
For $u,v = 1,\ldots,p$ and $t \in \mathbb{Z}$, $s \in \{t - \tau,\ldots,t\}$, the following are equivalent:
\begin{enumerate}
    \item[(1)] $X_{u,s} \notin \text{pa}(v,t)$.
    \item[(2)] $X_{v,t}$ and $X_{u,s}$ are non-adjacent in $G$.
    \item[(3)] $X_{v,t} \ind X_{u,s} \mid \bm{X}_S$ for some $S \subseteq \{(d,r): d = 1,\ldots,p;\ r = t - 2\tau,\ldots,t\}$.
\end{enumerate}
\end{lemma}

\noindent
\textit{Proof provided in Supplementary Section~\ref{appsec:theo_guar}.}

We recall that due to the $\tau$-order Markov property, all elements of $\text{pa}(v,t)$ must be of the form $X_{u,s}$ for $u=1,\ldots,p$ and $s \in \{t-\tau,\ldots,t\}$. Therefore, by Lemma~\ref{lemma:concept}, we have:
\begin{align}\label{eq:prop_adj}
X_{u,s} \not\ind X_{v,t} \mid \bm{X}_S 
\quad \text{for some } 
S \subseteq \{(d,r): d = 1,\ldots,p;\ r = t - 2\tau,\ldots,t\}.
\end{align}

We leverage this key property to formulate the Causal Inference in Time Series (CITS) algorithm for estimating $G$ and its rolled version $G_R$, and establish its theoretical guarantees in the subsequent sections. %We \textcolor{magenta}{later} show that when having no concurrent effects, assuming regularity conditions, the CITS algorithm achieves \textcolor{magenta}{\sout{the convergence rate for } consistency in} estimating $G$ and thereby $G_R$. \textcolor{magenta}{\sout{The following lemma is proved in Appendix \ref{proofLem3.1}.}}

%We use this key property to estimate $G$ by estimating specific conditional dependencies through statistical tests. We then use time order to direct the estimated skeleton since edges in $G$ can be directed by time order \textcolor{magenta}{being devoid of concurrent effects} \textcolor{green}{last part not written correctly}. This is put together in the Causal Inference in Time Series (CITS) algorithm in the following section. We \textcolor{magenta}{later} show that when having no concurrent effects, assuming regularity conditions, the CITS algorithm achieves \textcolor{magenta}{\sout{the convergence rate for } consistency in} estimating $G$ and thereby $G_R$. \textcolor{magenta}{\sout{The following lemma is proved in Appendix \ref{proofLem3.1}.}}

\subsection{Causal Inference in Time Series (CITS) Algorithm} 
\subsubsection{Oracle Version: With Conditional Independence Access}
The properties described in \textbf{Section~\ref{sec:properties}} (see Lemma~\ref{lemma:concept}) motivate a two-level strategy for identifying $\text{pa}(v,t)$ using a conditional independence oracle. First, for each $u = 1,\ldots,p$ and $s \in {t-\tau,\ldots,t}$, we test whether $X_{u,s}$ satisfies property~\eqref{eq:prop_adj}. Second, for each such pair $(u,s)$, we search over subsets $S$ of ${(d,r): d=1,\ldots,p;\ r= t-2\tau,\ldots,t}$ to determine if conditional independence holds for at least one such $S$.

If such a subset $S$ exists, the edge is deleted. Otherwise, the edge is retained. Applying this logic to time $t = 2\tau+1$, we identify all parents $\text{pa}(v,2\tau+1)$, from which $G$ is constructed. The rolled graph $G_R$ is then obtained by summarizing across $s = \tau+1,\ldots,2\tau$. 

The following result guarantees the correctness of the CITS-Oracle procedure under standard assumptions.

\begin{theorem}[CITS-Oracle Recovery Guarantee]\label{thm:CITS_oracle}
Let $\{\bm{X}_t\}_{t\in \mathbb{Z}}$ be a strictly stationary Markovian process of order $\tau$, and assume it can be uniquely represented by a time-invariant structural causal model (SCM) with DAG $G$ as in~\eqref{eqdef: process}, satisfying the faithfulness assumption. Then, the CITS-Oracle algorithm recovers:
\begin{itemize}
    \item all non-concurrent directed edges in $G$ exactly, and
    \item all concurrent edges in $G$ up to their Markov equivalence class.
\end{itemize}
\end{theorem}

\noindent
\textit{Proof provided in Supplementary Section~\ref{appsec:theo_guar}.}

\begin{remark}[Interpretation of Rolled Graph under Concurrent Effects]
The rolled graph $G_R$ is formed by collapsing the time-unrolled DAG $G$ into a directed graph over neuron indices. When $X_{u,t}$ and $X_{v,t}$ are adjacent in $G$ with $s = t$, the direction of interaction is identifiable up to its Markov Equivalence Class. If the Markov equivalence class has both $X_{u,t} \rightarrow X_{v,t}$ and $X_{v,t} \rightarrow X_{u,t}$, then both $u \rightarrow v$ and $v \rightarrow u$ will appear in $G_R$. This reflects the nonidentifiability of direction under Markov equivalence and is consistent with representations used in constraint-based discovery methods.
\end{remark}

\begin{algorithm}[t!]
\SetKwInOut{Input}{Input}
\SetKwInOut{Output}{Output}
\Input{$X_{v,t}, v=1,\ldots,p; t=1,\ldots,2\tau+1$, Conditional Independence Information.}
\Output{DAG $G$ and Rolled DAG $G_R$}

Start with an initial DAG $G_1$ between nodes $\{X_{v,t}: v=1,\ldots,p; t=1,\ldots,2\tau+1\}$ with edges $X_{u,s}\rightarrow X_{v,2\tau+1}$ for $s=\tau+1,\ldots,2\tau+1$, $u,v=1,\ldots,p$.

\Repeat{all $u,v=1,\ldots,p$, $s=\tau+1,\ldots,2\tau+1$ are tested.}{
\Repeat{edge $X_{u,s}\rightarrow X_{v,2\tau+1}$ is deleted or all $S\subseteq\{(d,r):d=1,\ldots,p; r=1,\ldots,2\tau+1\}\setminus$\\$\{(u,s),(v,2\tau+1)\}$ are selected.}{
            Choose $S\subseteq\{(d,r):d=1,\ldots,p; r=1,\ldots,2\tau+1\}$.
            
            \uIf{$X_{u,s}\ind X_{v,2\tau+1} ~\vert~ \bm{X}_{S}$}{
                
                Delete edge $X_{u,s}\rightarrow X_{v,t}$.
                
                Denote this new graph by $G_1$.
}
}
}

$\text{pa}(v,2\tau+1) = \{X_{u,s}: X_{u,s}\rightarrow X_{v,2\tau+1} \text{ in $G_1$}; s=\tau+1,\ldots,2\tau+1; u=1,\ldots,p\}$.

Obtain the DAG $G$ by edges directing from $\text{pa}(v,2\tau+1)\rightarrow X_{v,2\tau+1}$.

Obtain the Rolled Graph $G_R$ with nodes $v=1,\ldots,p$ and edge $u\rightarrow v$ if $X_{u,s}\in \text{pa}(v,2\tau+1)$ for some $u=1,\ldots,p; s=\tau+1,\ldots,2\tau+1$.

\caption{CITS-Oracle}\label{alg:CITS_oracle}
\end{algorithm}

\subsubsection{Practical Implementation: Sample-Based CITS}
For the sample version of the CITS algorithm (CITS-sample), we replace the conditional independence statements by outcomes of statistical tests for conditional dependence based on a sample. For details of appropriate conditional dependence tests, \textbf{see Supplementary \ref{appsec:condindtests}}. Note that our method assumes access to only a single realization of the stochastic process, as is common in practice. To address this, we construct time-windowed samples by taking consecutive time windows of a duration of $2\tau+1$ (see \textbf{Fig. \ref{fig:graphexample}B}). That is, the samples are $\chi_k = \{X_{v,t}:v=1,\ldots,p; t= (2\tau+1)(k-1)+1,\ldots,(2\tau+1)k\}, k=1,\ldots,N$ where $N=\floor{\frac{n}{2\tau+1}}$. For example, for some $v=1,\ldots,p$ and $t=1,\ldots,2\tau+1$, the samples for $X_{v,t}$ based on $\{\chi_k\}_{k=1}^N$ are: $\{X_{v,(2\tau+1)(k-1)+t}:k=1,\ldots,N\}$. Then for $u,v=1,\ldots,p; t=2\tau+1; s=\tau+1,\ldots,2\tau+1; S\subset\{(d,r):d=1,\ldots,p; r= 1,\ldots,2\tau+1\}$, we perform statistical tests of the form $X_{u,s}\ind X_{v,t} \vert \bm{X}_{S}$ based on samples $\chi_k$. We then estimate $\text{pa}(v,2\tau+1)$ using the same steps as CITS-oracle, but replacing the conditional independence statements by the outcome of the statistical tests. The CITS-sample algorithm is outlined in Algorithm \ref{alg:CITS_sample}.

\begin{algorithm}[t!]
\SetKwInOut{Input}{Input}
\SetKwInOut{Output}{Output}
\Input{$X_{v,t}, v=1,\ldots,p; t=1,\ldots,n$}
\Output{Estimated DAG $\hat{G}$ and Rolled Graph $\hat{G}_R$}

Construct Time-Windowed samples: $\chi_k = \{X_{v,t}:v=1,\ldots,p; t= (2\tau+1)(k-1)+1,\ldots,(2\tau+1)k\}, k=1,\ldots,N$.

Run the CITS-Oracle algorithm while replacing the conditional independence statement in Line 5 by statistical tests \textbf{in Supplementary \ref{appsec:condindtests}} based on samples $\chi_k$ to output DAG $\hat{G}$ and Rolled Graph $\hat{G}_R$.

\caption{CITS-Sample}\label{alg:CITS_sample}
\end{algorithm}

\subsection{Assigning Edge Weights in the Inferred Graphs}\label{sec:edgeweights}
In this section, we first assign an edge weight, denoted $w_{u,s}^{v,t}$, for the edge $X_{u,s}\rightarrow X_{v,t}$ in the Unrolled DAG estimate $\hat{G}$ obtained by CITS.
\begin{equation}\label{eq:causaleffectgaussian}
    w_{u,s}^{v,t}=\left\{\begin{array}{lr}
    0, & \text{ if } X_{u,s} \not\in pa_{\hat{G}}(X_{v,t}),\\
    &\\
    \text{coefficient of }X_{u,s} \text{ in }&\\
    X_{v,t} \sim pa_{\hat{G}}(X_{v,t}) & \text{ if } X_{u,s} \in pa_{\hat{G}}(X_{v,t})\end{array}\right.
\end{equation}
where $X_{v, t} \sim pa_{\hat{G}}(X_{v,t})$ is a shorthand for linear regression of $X_{v,t}$ on $pa_{\hat{G}}(X_{v,t}) = \{X_{a,b} : X_{a,b} \rightarrow X_{v,t} \text{ is an edge in } \hat{G}\}$. That is, the edge weights are obtained based on linear regressions corresponding to the estimated DAG of the SCM.

Next, we assign an edge \textit{weight} for the edge $u\rightarrow v$ in the Rolled graph estimate $\hat{G}_R$, following the way $\hat{G}_{R}$ is defined from $\hat{G}$ in the CITS algorithm: the edge weight for $u \to v$ in $\hat{G}R$, denoted $w_u^v$, is defined as the average of the weights $w{u,s}^{v,2\tau+1}$ over all $s \in {\tau+1,\ldots,2\tau}$ such that $X_{u,s} \in \text{pa}(X_{v,2\tau+1})$ in $\hat{G}$.

\subsection{Theoretical Guarantees}
We establish that the CITS algorithm is statistically consistent under mild assumptions on the underlying time series. Specifically, CITS recovers the true time-unrolled causal graph for non-concurrent edges and up to Markov equivalence for concurrent edges, and similarly recovers the rolled graph structure. This result holds under stationarity, faithfulness, and consistent conditional dependence testing. Such consistent conditional dependence tests in the time series setting include Fisher's partial correlation-based test for Gaussian distributed samples and Hilbert-Schmidt Independence Criterion, which is distribution-free. 

More precisely, let $\mu(X_{u,s},X_{v,t}\vert\bm{X}_S)$ denote a measure of conditional dependence of $X_{u,s}$ and $X_{v,t}$ given $\bm{X}_S$, i.e. it takes value $0$ if and only if we have conditional independence, and let $\hat{\mu}_n(X_{u,s},X_{v,t}\vert\bm{X}_S)$ be its consistent estimator. In Theorem \ref{thm:main}, we will use $\hat{\mu}_n(X_{u,s},X_{v,t}\vert\bm{X}_S)$ to construct a conditional dependence test guaranteeing consistency of the CITS estimate.

% \paragraph{Assumption 1}\label{assumption2} \begin{align*}\inf \{\vert\mu(X_{u,s},X_{v,t}\vert\bm{X}_S)\vert : \mu(X_{u,s},X_{v,t}\vert\bm{X}_S)\neq 0, u,v=1,\ldots,p;s=\tau+1,\ldots,2\tau+1;\\
%     \quad S\subseteq \{(d,r):d=1,\ldots,p;r=1,\ldots,2\tau+1\}\setminus\{(u,s),(v,2\tau+1)\}\}>0.\end{align*} 

\begin{theorem}[CITS Consistency for Time-Unrolled and Rolled Graphs]\label{thm:main}
Let $\{\bm{X}_t\}_{t\in \mathbb{Z}}$ be a strictly stationary stochastic process of finite Markov order $\tau$ that follows a time-invariant structural causal model with DAG $G$ as in \eqref{eqdef: process}, and assume the distribution is faithful to $G$. Let $\hat{G}_n$ and $\hat{G}_{n,R}$ denote the estimated unrolled and rolled graphs, respectively, obtained by the sample CITS algorithm using a consistent conditional dependence test:
\[
\left\vert\hat{\mu}_n(X_{u,s},X_{v,t} \mid \bm{X}_S)\right\vert > \gamma
\]
for some fixed $\gamma > 0$, where $\hat{\mu}_n$ consistently estimates a valid conditional dependence measure.

Then, as $n \rightarrow \infty$:
\begin{itemize}
    \item $\hat{G}_n$ recovers the true time-unrolled DAG $G$ up to the Markov equivalence class over concurrent edges (i.e., for edges with $s = t$), and exactly for all non-concurrent edges (i.e., $s < t$).
    \item $\hat{G}_{n,R}$ recovers the rolled graph $G_R$ with all non-concurrent edges recovered exactly, and concurrent edges recovered up to their Markov equivalence class.
\end{itemize}
\end{theorem}

\noindent
\textit{Proof provided in Supplementary Section~\ref{appsec:theo_guar}.}

\begin{corollary}\label{cor:latent}
Under the assumptions of Theorem~\ref{thm:main}, the CITS algorithm consistently recovers the correct lagged causal graph among observed variables, even when the underlying structural causal model includes latent variables, provided that:
\begin{itemize}
    \item[(i)] The process is one-step Markovian;
    \item[(ii)] All directed effects among observables occur strictly across time steps (i.e., no concurrent edges);
    \item[(iii)] The marginal distribution over observed variables is faithful to the induced causal graph.
\end{itemize}
\end{corollary}

\paragraph{Discussion.}
This corollary follows from the temporal structure of the model. When all causal effects among observables are strictly time-lagged (i.e., from $X_{u,t-1}$ to $X_{v,t}$), with no concurrent (within-time-slice) effects, and the system is one-step Markovian, latent variables can induce at most contemporaneous dependencies within a single time step. Since CITS only tests for lagged conditional dependencies and conditions on past variables, it avoids inferring spurious lagged edges due to latent confounding under these assumptions. As long as the marginal distribution over observed variables is faithful to the induced causal graph, the consistency guarantees of Theorem~\ref{thm:main} continue to apply. The assumption that the marginal distribution over observed variables is faithful to the marginal causal graph is standard in causal discovery. While faithfulness is not guaranteed under marginalization, since conditional independencies can arise through parameter cancellations, it holds outside a measure-zero subset of structural equation models~\cite{spirtes2000causation}. Hence, we adopt it here as a generic condition that enables consistent recovery of lagged causal structure in the presence of latent variables.

This result is consistent with prior findings in the causal time series literature, where one-step Markovianity and temporal separation of effects reduce the risk of confounding in lagged causal inference~\cite{lindgren2021identifiability}.

%Moreover, CITS remains robust to latent confounding when the process is one-step Markovian and exhibits no concurrent effects, a scenario common in high-resolution neural recordings. 
%See \textbf{Supplementary \ref{appsec:theo_guar}} for the full theorem statements, assumptions, and proofs.

\section{CITS Outperforms Existing Methods on Simulated Datasets}

\begin{figure}[t!]
    \centering
    \includegraphics[width=\textwidth]{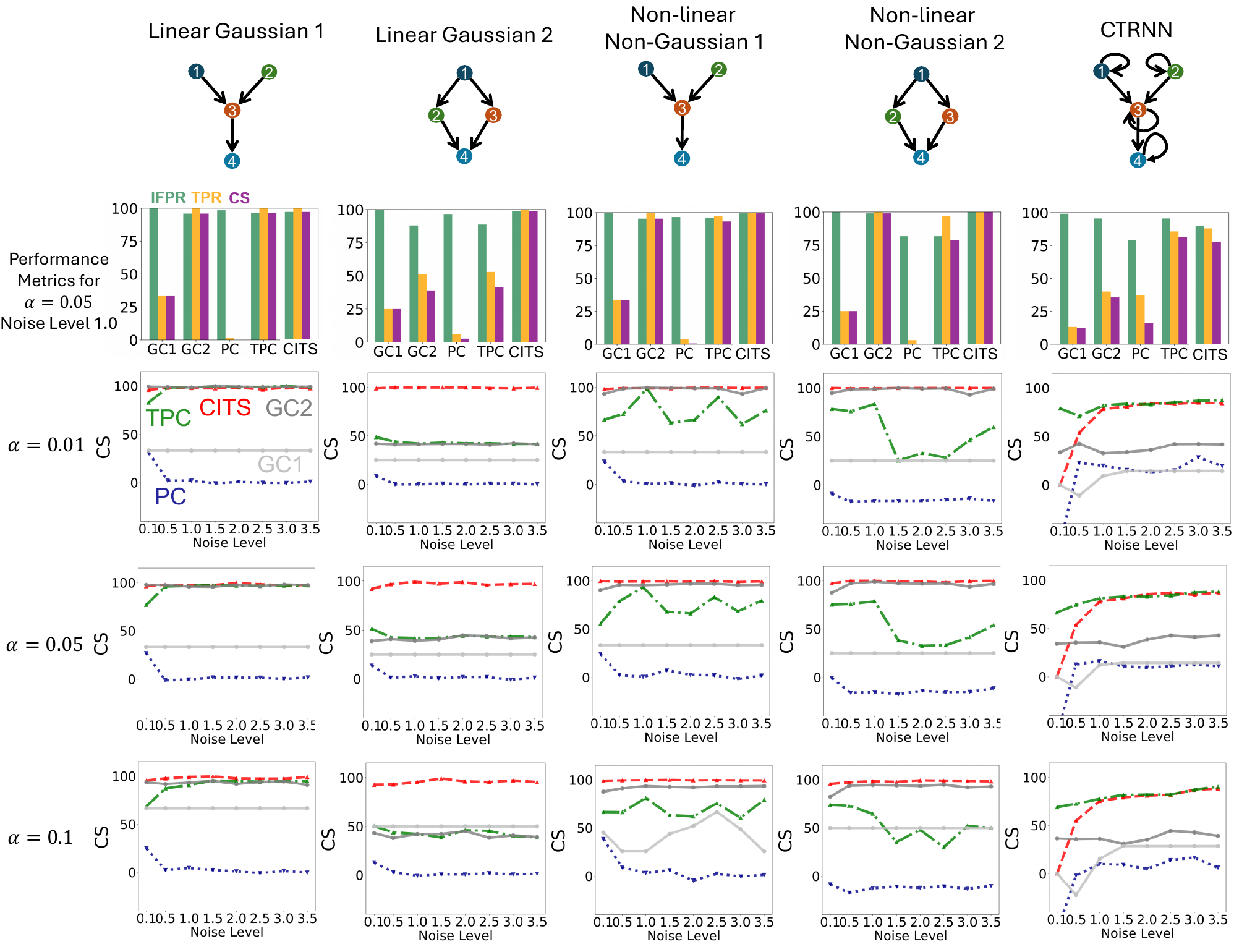}
    \caption{\textbf{Comparative study of inferring the Rolled graph.} Inference of Rolled graph for five simulation settings (left to right): Linear Gaussian Models 1 and 2, Non-linear Non Gaussian Models 1 and 2, Continuous Time Recurrent Neural Network (CTRNN). Row 1: The ground truth for each simulation paradigm is graphically represented. The performances of the five methods Granger Causality 1 (GC1), Granger Causality 2 (GC2), Peter-Clark (PC), Time-Aware PC (TPC) and CITS, are shown in terms of three metrics (right column): $1 -$ False Positive Rate (IFPR) (green), True Positive Rate (TPR) (orange) and Combined Score (CS) (purple). Row 2 shows the performance the three metrics for $alpha = 0.05$ and noise level $1.0$. The CS of the methods over varying noise levels in simulation $\eta = 0.1,0.5,1.0,\ldots,3.5$, with significance level $\alpha = 0.01, 0.05, 0.1$ are also demonstrated in rows 3-5 respectively.}
    \label{fig:compsim}
\end{figure}

We compare the performance of CITS, Pairwise Granger Causality (GC1), Multivariate Granger Causality (GC2), naive application of PC algorithm (PC), and Time-Aware PC algorithm (TPC), to recover the ground truth causal relations in simulated datasets (see \textbf{Supplementary \ref{appsec:cinf_tsreview}}). We use simulated datasets from a variety of time series models ranging from linear to non-linear models, with and without common causes, and consider both the Gaussian and non-Gaussian noise settings (See \textbf{Supplementary \ref{appsec:simulstudy}}). In the simulations, for each model, we obtain 25 simulations of the entire time series each for different noise levels $\eta \in \{0.1,0.5,1,1.5,2,2.5,3,3.5\}$. All the time series simulated have $n=1000$ samples. We also used the level $\alpha$ of the conditional dependence test with $\alpha$ ranging in $0.01, 0.05$ and $0.1$. The performance of the methods in recovering the ground truth causal relationships is summarized using the following three metrics: (1) Combined Score (CS), (2) True Positive Rate (TPR) and (3) 1 - False Positive Rate (IFPR) (see Methods at \textbf{Section~\ref{sec:methods_metrics}} for definitions). Conditional independence tests are based on partial correlations in Gaussian settings and the Hilbert-Schmidt criterion in non-Gaussian settings (see \textbf{Section \ref{sec:methods_metrics}}).

In each of our simulation settings, there are $4$ neurons and $16$ possible edges (including self-loops), leading to a total of $400$ possible edges across $25$ simulations. CITS and TPC show superior performance in recovering the true Rolled graph at noise level $\eta = 1$ and thresholding parameter $\alpha = 0.05$ (\textbf{Fig.~\ref{fig:compsim}}). Their advantage is also evident in terms of Combined Score (CS) across varying noise levels $\eta$ and significance levels $\alpha$ in each simulation setting (\textbf{Fig.~\ref{fig:compsim}}). CITS consistently shows the strongest performance across all simulation settings, with TPC closely following. In the \textit{Linear Gaussian 1} scenario, where the true graph has three converging edges, CITS and TPC both achieve a TPR of 100\% with high Combined Scores (97.2\% and 96.6\%, respectively), while GC1, despite having an IFPR of 100\%, detects only one of the three true edges (TPR = 33.3\%). In \textit{Linear Gaussian 2}, which includes both a common cause and a common effect, CITS again leads with a CS of 99\%, followed by TPC (41.7\%) and GC2 (39\%). Here, GC2’s performance is hindered by multicollinearity in regression, and PC shows poor sensitivity with a TPR of just 6\%. In the \textit{Non-linear Non-Gaussian 1} scenario, CITS achieves a CS of 99.4\%, outperforming all methods, while TPC follows closely (93.3\%). GC2 detects all true edges (TPR = 100\%) but introduces slightly more false positives. In the \textit{Non-linear Non-Gaussian 2} case, CITS again attains perfect scores across metrics (TPR = 100\%, IFPR = 100\%, CS = 100\%). Finally, in the \textit{CTRNN} setting with self-loops, CITS and TPC are the only methods that reliably detect recurrent structures, achieving TPRs of 88\% and 85.7\%, respectively. TPC achieves the highest CS (81.3\%), just ahead of CITS (77.8\%). Overall, CITS proves to be the most robust method across both linear and nonlinear settings, maintaining a high detection rate while minimizing false positives.

We then performed a systematic comparison of the Combined Score for CITS and other methods across noise levels $\eta$ ranging from $0.1$ to $3.5$ and significance levels $\alpha = 0.01,0.05,0.1$ (\textbf{Fig. \ref{fig:compsim}}). In the Linear Gaussian 1 scenario, CITS has a CS of nearly $100\%$ across noise levels greater than $1.0$ and all signifiance levels considered, closely matching the parametric GC2 model as well as non-parametric TPC, which are followed by GC1 in performance and lastly PC. In the Linear Gaussian 2 scenario, the distinction is remarkable, where CITS has a CS of $\approx 100\%$ across all levels of simulation noise and significance level $\alpha$, however the next best model for this setting are TPC and GC2 with a CS of $\approx 50\%$, followed by GC1 and lastly PC. In the Non-linear Non-Gaussian 1 and 2 scenarios, CITS has the highest CS compared to other methods across levels of noise and $\alpha$. In the CTRNN scenario, the best CS achieved is lower compared to other simulation paradigms. However, CITS and TPC have better performance compared to the other methods for noise level $\eta\geq 0.5$ and all $\alpha$.

\begin{figure}[t]
    \centering
    \includegraphics[width = 0.9\textwidth]{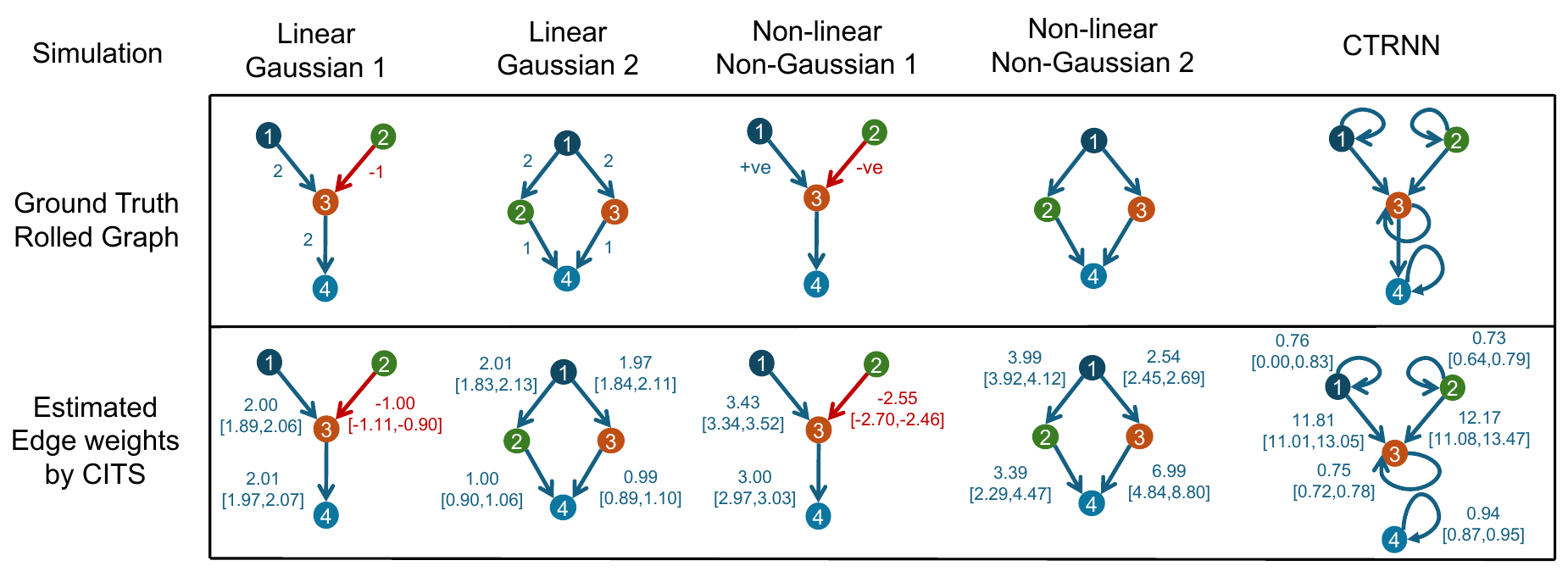}
    \caption{\textbf{Comparison of Ground Truth and Estimated Causal Edge Weights Across Simulation Paradigms} Top row: Ground truth edges for the simulation paradigms of Linear Gaussian 1, Linear Gaussian 2, Non-linear Non-Gaussian 1, Non-linear Non-Gaussian 2, and Continuous Time Recurrent Neural Network (CTRNN) (left to right). The ground truth edge weights are well-defined for linear paradigms. Bottom row: Estimated edge weights (median [min, max]).}
    \label{fig:edge_wts}
\end{figure}

The inferred edge weights and their directional nature (positive/increasing or negative/decreasing) at noise level $\eta = 1$ and threshold $\alpha = 0.05$ reveal that CITS most accurately recovers both the magnitude and sign of causal influences across scenarios (\textbf{Fig.~\ref{fig:edge_wts}}). Across all simulation settings, the edge weights inferred by CITS closely match the true values in both magnitude and sign. In the \textit{Linear Gaussian 1} scenario, the median estimates for $1 \rightarrow 3$, $2 \rightarrow 3$, and $3 \rightarrow 4$ are nearly identical to the ground truth values (2, –1, and 2), with narrow ranges across trials and consistent sign recovery. Similarly, in \textit{Linear Gaussian 2}, all four edges—including those involving common causes and effects—are accurately estimated, with medians near 2 and 1, and signs matching the true positive directions in every case.

In the \textit{Non-linear Non-Gaussian} scenarios, despite the presence of sinusoidal dependencies and wider value ranges, CITS maintains high fidelity in edge weight estimation. Median values for connections such as $1 \rightarrow 3$ and $3 \rightarrow 4$ remain tightly centered around 3 to 3.5, while signs consistently reflect the correct increasing or decreasing influences. In the \textit{CTRNN} setting, CITS detects self-loops with moderate strength (e.g., 0.73 to 0.94) and identifies strong feedforward edges such as $1 \rightarrow 3$ and $2 \rightarrow 3$ with high estimated weights (around 12). In all scenarios, the directional nature of each edge is recovered correctly across trials, highlighting the robustness of CITS to structural complexity and nonlinearity.

\section{CITS Reveals Statistically Causal Neural Circuitry in the Mouse Brain During Visual Tasks}\label{sec:application}

In neuroscience, the term Functional Connectivity (FC) refers to the network of interactions between individual units of the brain, such as neurons or brain regions, with respect to their activity over time \cite{reid2012functional,biswas2022statistical1}. The main purpose of identifying the FC is to gain an understanding of how neurons work together to create brain function. The FC can be represented as a graph, where nodes denote neurons and edges denote a stochastic relationship between the activities of connected neurons. These edges can be undirected, indicating a stochastic association, whence the FC is termed as Associative FC (AFC). Alternatively, the edges can be directed and represent a statistically causal relationship between the activities of neurons, whence the FC is termed as Causal FC (CFC). Finding the CFC is expected to facilitate the inference of the governing neural interaction pathways essential for brain functioning and behavior \cite{finn2015functional,biswas2022statistical1}, and serves as a promising biomarker for neuro-psychiatric diseases \cite{nakamura2017early}. The CFC is represented by a directed graph whose nodes are the neuron labels, and has an edge from neuron $u\rightarrow v$ if the activity of neuron $u$ at time $t$ has a statistically causal influence on the activity of neuron $v$ at a later time $t'$ \cite{reid2019advancing,biswas2022statistical1,biswas2022statistical2}. In terms of the framework described in \textbf{Section \ref{sec:tsscm}}, the CFC can be represented by the Rolled Graph of statistically causal neural interaction in the neural time series. Thereby, the CFC can be inferred by CITS algorithm (See Definition \ref{def:rolledgraph}).In this section, we show how CITS reveals putative CFCs between neurons in the mouse brain during visual tasks. To do so, we analyzed high‑resolution, large‑scale electrophysiological recordings from the Allen Institute’s Visual Coding Neuropixels dataset in the Allen Brain Observatory~\cite{allenbrainobs}. The dataset includes sorted spike trains recorded simultaneously across six cortical visual areas, hippocampus, thalamus, and other adjacent structures of the mouse brain, while being presented with different types of visual stimuli. For our CFC analysis we chose responses to natural scenes, static gratings andGabor patches. We omitted responses to full-field flash stimuli from the analysis because such stimuli can evoke widespread, near-synchronous responses across the brain, leading to strong, stimulus-locked common inputs that inflate statistical dependencies and may masquerade as direct causal interactions. Prior work highlights how such extrinsic correlations, driven by the stimulus ensemble, can overwhelm intrinsic connectivity patterns~\cite{okun2009instantaneous,nonnenmacher2018intrinsic}. The visual stimuli are repeatedly presented to the mice and the data is recorded using the Neuropixels technology, consisting of multiple electrodes inserted into the brain allowing real-time recording from hundreds of neurons across brain different brain regions (For more details, see in \textbf{Methods Section \ref{sec:datadescriptionfull}}.)

\subsection{Comparison of CITS and Other FC Methods.}

\begin{figure}[t!]
    \centering
    \includegraphics[width = \textwidth]{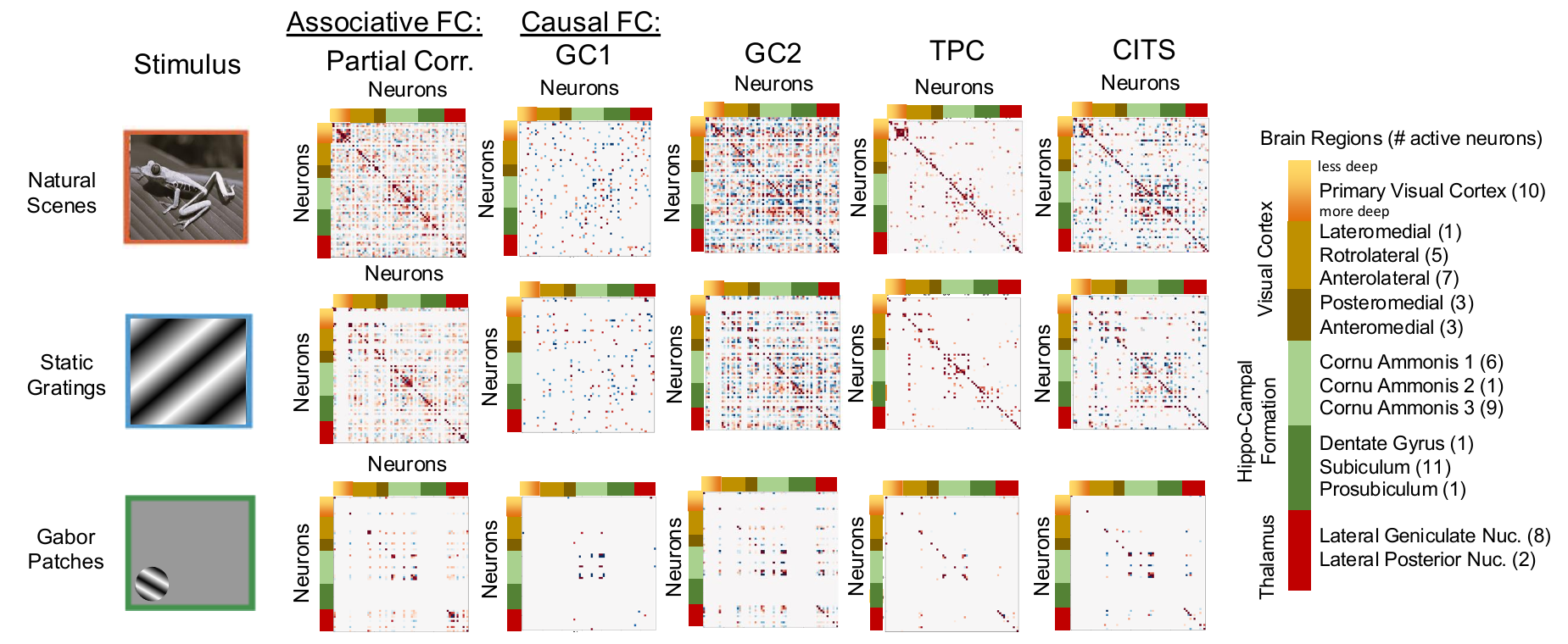}
    \caption{\textbf{Comparison of Associative and Causal Functional Connectivity Methods on Neuropixels Mouse Brain Data.} Four different methods for inferring functional connectivity (FC) were compared using benchmark mouse brain data from the Allen Institute's Neuropixels dataset. These methods include associative FC using Partial Correlation, and causal FC using Granger Causality 1 (GC1), Granger Causality 2 (GC2), Time-Aware PC (TPC), and CITS. 
    FC is estimated and visualized as an adjacency matrix with edge weights, which is symmetric for associative methods and asymmetric for causal methods. In each matrix, a non-zero entry at position $(i, j)$ indicates a directed connection from neuron $i$ to neuron $j$.}
    \label{fig:resneuropixels}
\end{figure}

We evaluated CITS alongside Time-Aware PC (TPC), Granger Causality (GC1 and GC2) and Partial Correlation which are well-known methods to obtain CFC and Associative Functional Connectivity (AFC) from electrophysiological recordings (\textbf{Fig. \ref{fig:resneuropixels}}). The Causal FC (CFC) is expected to be a directed subset of the AFC and consistent with its overall AFC pattern \cite{wang2016efficient}. Such is observed in the CFC obtained by TPC and CITS, which yield asymmetric adjacency matrices, yet match the overall pattern in the AFC in a sparse and dense manner, respectively. In contrast, the CFC obtained by GC1 is sparse but does not match the patterns in the AFC well, while GC2 yields an overly dense graph. This is likely due to Granger causality's sensitivity to noise and weak correlations in neural recordings, which can result in spurious causal links being inferred even when no true dependency exists \cite{stokes2017problems,seth2015granger}. 
%In contrast, the CFC obtained by GC1 is sparse but does not match the patterns in the AFC so well, and the one obtained by GC2 is too dense, likely due to noise in neural measurements. 
In the CFC obtained by both TPC and CITS (\textbf{Fig. \ref{fig:resneuropixels}}), natural scenes evoke greater connectivity within active neurons in the Primary Visual Cortex, and static gratings evoke greater connectivity in the Posteromedial and Anteromedial Visual Cortex compared to other stimuli. All the three stimuli exhibit distinct connectivity patterns in the Cornu Ammonis regions of the Hippocampal Formation. In addition, natural scenes and static gratings induce more prominent connectivity within the Subiculum compared to other stimuli.

These findings are consistent with known functional specializations in the mouse brain. The increased connectivity within the Primary Visual Cortex under natural scenes aligns with prior studies showing that natural stimuli evoke richer and more spatially diverse activation patterns in early visual areas \cite{froudarakis2014population}. The enhanced connectivity in posteromedial and anteromedial visual areas under static gratings likely reflects their tuning to orientation and spatial frequency \cite{andermann2011functional}. Interestingly, all three stimuli elicit distinct patterns in the Cornu Ammonis regions of the hippocampus, and natural scenes and gratings induce stronger connectivity within the Subiculum. These effects may reflect stimulus-specific modulation of hippocampal circuits involved in contextual learning and memory integration \cite{igarashi2014coordination}. Together, these results suggest that CITS and TPC effectively capture biologically meaningful, stimulus-dependent causal dynamics across cortical and subcortical structures.

\begin{figure}[t!]
    \centering
    \includegraphics[width = 0.8\textwidth]{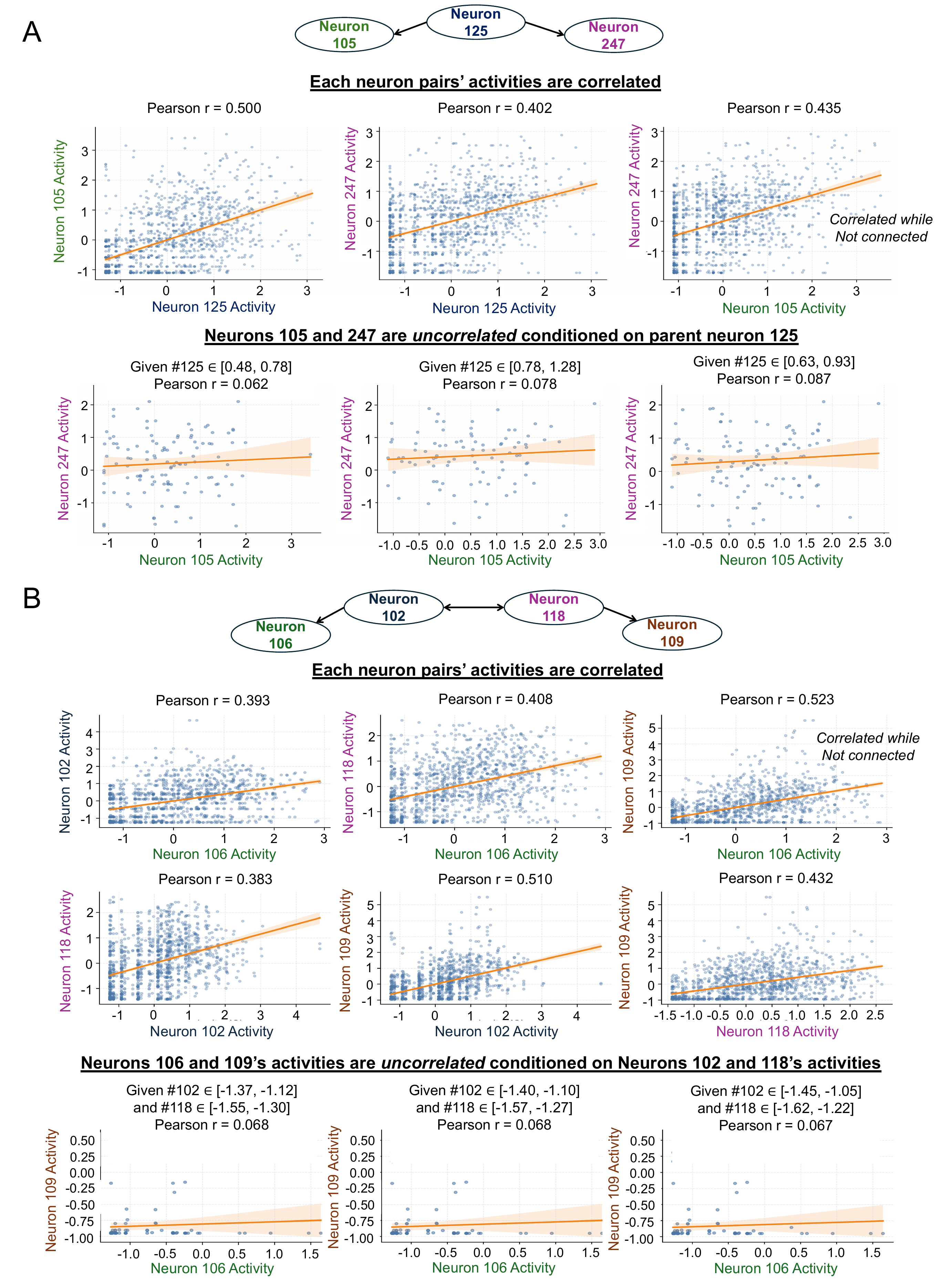}
    % \vspace{-0.7cm}
    \caption{\textbf{Examples of Neural Signal Dependencies within Inferred CFC Motifs. A.} Common-Source neuronal motif identified by CITS. In the inferred CFC, for the motif of neuron 105 $\leftarrow$ 125 $\rightarrow$ 247, pairwise plots exhibit correlations between all pairs including spurious correlation between 105 and 247 (red) even though not directly connected in the motif (top row). As predicited, conditioned on the parent, 105 and 125's signals are uncorrelated. \textbf{B.} Multi-Parent neuronal motif identified by CITS. In a more complex motif of 106 $\leftarrow$ 105 $\leftrightarrow$ 118 $\rightarrow$ 109, pairwise plots exhibit correlations between all pairs (rows 1-2). As predicted, conditioned on both parents, neurons 106 and 109 become uncorrelated. This highlights CITS ability in correctly identifying these nodes as being not directly connected.}

    % Small Motifs in the Inferred  Causal Functional Connectivity and their Neural Signal Dependencies.} A. In the inferred CFC, for the motif of neuron 105 $\leftarrow$ 125 $\rightarrow$ 247, pairwise plots exhibit correlations between all pairs including spurious correlation between 105 and 247 (red) even though not directly connected in the motif (top row). As expected, conditioned on the parent in the motif 125, 105 and 125's signals are uncorrelated. B. For the motif of neuron 106 $\leftarrow$ 105 $\leftrightarrow$ 118 $\rightarrow$ 109, pairwise plots exhibit correlations between all pairs (rows 1-2), as well as when conditioned on neuron 102 (row 3) or 118 (row 4) individually. However, conditioned on both neurons 102 and 118, neuron 106 and 109 stand uncorrelated.}
    \label{fig:corrplots}
\end{figure}

\subsection{Neural Signal Dependencies within Inferred CFC Motifs}

To interpret and validate the inferred CFC, we examined whether the neural signals obey the conditional independence patterns predicted by the graph structure. A defining feature of CFC, and what distinguishes it from standard functional connectivity, is its ability to isolate direct interactions from indirect network effects. Theoretically, this relies on the directed Markov property: two nodes may appear highly correlated due to a shared ancestor, but this dependence is rendered spurious (statistically independent) once the activity of the parent nodes is controlled.

We tested this phenomena by analyzing specific motifs within the inferred network. As shown in \textbf{Fig.~\ref{fig:corrplots}}, marginally correlated neuron pairs often lack a direct connection in the CFC graph, implying their relationship is indirect. For example, in a "common source" motif (\textbf{Fig.~\ref{fig:corrplots}A}), Neurons 105 and 247 exhibit a strong pairwise correlation (Pearson $r \approx 0.435$) despite not being connected. Consistent with the inferred CFC structure, conditioning on their shared parent (Neuron 125) explains away this relationship, reducing the residual correlation to near zero ($r \approx 0.06-0.08$).

This validation extends to more complex topologies involving multiple common drivers. In the motif shown in \textbf{Fig.~\ref{fig:corrplots}B}, Neurons 106 and 109 share two distinct parents (Neurons 102 and 118) but have no direct link between them. While their raw activities are highly correlated (r=0.523), the CFC model correctly predicts that this is an indirect association. Indeed, when both parent neurons are conditioned on simultaneously, the correlation between the children vanishes ($r \approx 0.068$). These examples highlight the ability of CITS in identifying associations driven by upstream propagation and thereby providing a more rigorous map of neural information flow.

% To interpret the inferred causal functional connectivity (CFC), we examined the neural signal relationships within small connectivity motifs. Although every pair of nodes in the motifs exhibit a correlation, each motif implies specific conditional independence patterns: neurons that appear correlated because of a shared parent should become uncorrelated once that parent is conditioned on, and neurons that share multiple parents may only exhibit independence when the full parent set is controlled, and more broadly, satisfying the directed Markov property \cite{pearl2009causality,biswas2022statistical2}. This implicates the relationship in the inferred CFC motif. By plotting both marginal pairwise correlations and conditioned signal relationships for selected motifs, we assess whether the observed neural activity matches these causal predictions. The examples in \textbf{Fig.~\ref{fig:corrplots}} show that the inferred CFC captures these neural dependencies accurately, since spurious correlations disappear under the correct conditioning while genuine links remain. This provides direct justification for the structure obtained in the inferred graph.

\subsection{Stimulus-Specific Causal Functional Connectivity.}

\begin{figure}[t!]
    \centering
    \includegraphics[width = \textwidth]{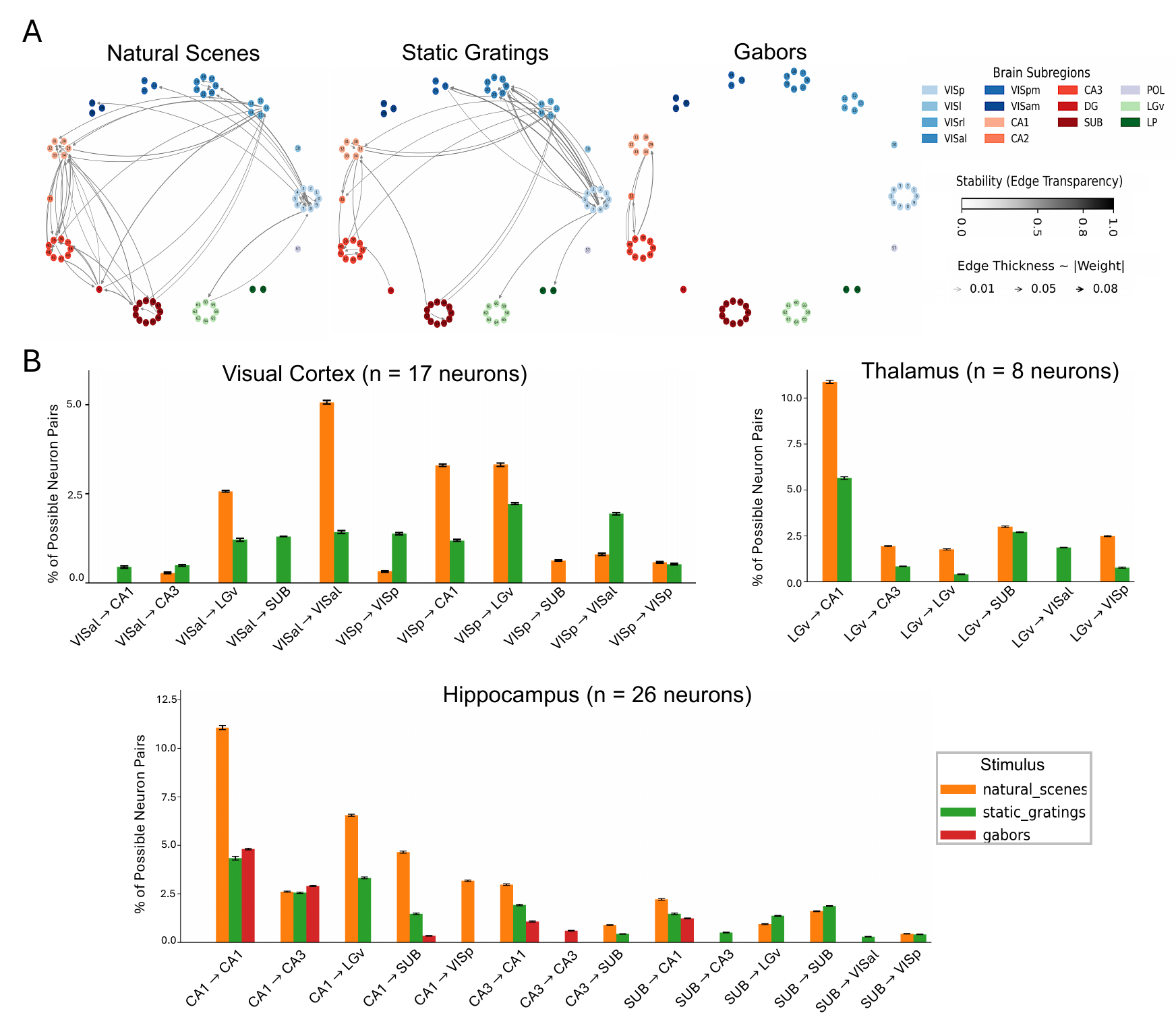}
    \caption{\textbf{Visual Stimulus-Specific Stable Causal Functional Connectivity and its Region-to-Region Distribution in the Mouse Brain.} A. Representative Causal Functional Connectivity (CFC) graphs for each stimulus type. Edges shown are those consistently present in over 80\% of trials within each stimulus category (natural scenes, static gratings, and Gabor patches). Edge thickness indicates connection weight, transparency indicates edge stability across trials, and nodes (brain regions) are colored according to anatomical groupings. 
    B. Anatomical distribution of statistically causal edges across stimulus conditions. Bar plots quantify the percentage of neuron pairs exhibiting directed connections between each source and target region, based on edges that passed the 80\% stability threshold. Regions are grouped into hippocampus (CA1, CA3, SUB), thalamus (LGv), and visual cortex (VISal, VISp). Bar heights indicate the percentage relative to all possible directed neuron pairs in each region group.}
    \label{fig:stimtypegraphs}
\end{figure}

To investigate how different visual stimuli modulate directed neural communication, we constructed representative Causal Functional Connectivity (CFC) graphs for each stimulus type: natural scenes, static gratings, and Gabor patches (\textbf{Fig.~\ref{fig:stimtypegraphs}A}). Augmented Dickey-Fuller (ADF) tests confirmed that 99.58\% of the 7068 neuron time series across all trials were stationary at a significance level of $p = 0.05$, supporting the suitability of applying CITS \cite{dickey1979distribution}. Each graph visualizes edges that were consistently inferred in over 80\% of trials. Edge thickness encodes average statistically causal strength, transparency reflects edge stability, and node color denotes anatomical grouping by brain region.  

\textit{Natural scenes} evoked the densest graph, comprising 79 edges spanning visual cortex (VISp, VISl, VISal), hippocampal subregions (CA1, CA2, SUB), and thalamic nuclei. The network included both intra-cortical and cortical-to-hippocampal links, consistent with the hypothesis that rich, complex stimuli engage distributed circuits for sensory integration and memory formation~\cite{saleem2018subcortical,siegle2021survey}. \textit{Static gratings} elicited 52 edges, mostly concentrated within visual areas, with relatively sparse cross-regional communication. This aligns with prior studies showing that structured, repetitive stimuli drive localized cortical processing~\cite{Niell2008}. \textit{Gabor patches} yielded only 5 edges, forming a minimal graph dominated by hippocampal interactions, reflecting constrained neural recruitment for highly simplistic visual inputs~\cite{andermann2011functional}. These findings demonstrate a systematic relationship between stimulus complexity and statistically causal network structure: stimuli with greater spatial richness induce more widespread and strongly coupled neural interactions, whereas simpler stimuli recruit localized, lower-density networks~\cite{siegle2021survey}.

\subsection{Anatomical Distribution of Causal Influence.}

To further examine the anatomical structure and stimulus specificity of the inferred interactions, we quantified the proportion of directed edges originating from and targeting each brain region across stimulus types (\textbf{Fig.~\ref{fig:stimtypegraphs}.B}). Natural scenes elicited the most widespread connectivity across all regions, including strong cortico-hippocampal and hippocampo-cortical interactions. VISp and VISal acted as dominant sources, sending projections to LGv, CA1, and SUB, while CA1 and SUB also showed recurrent intra-region connectivity. Static gratings produced a more localized pattern, with most edges confined to intra-visual cortex links and limited hippocampal involvement. Gabor stimuli yielded sparse connectivity overall, with weak interactions observed primarily within hippocampal subfields.

These patterns highlight how stimulus complexity shapes the anatomical footprint of causal communication in the brain. Natural scenes recruit integrative loops across sensory and memory circuits, whereas simpler stimuli yield more constrained, unimodal networks. The directed edge distribution reinforces the biological plausibility of the CITS-inferred architecture and supports the interpretation that causal pathways dynamically reconfigure based on stimulus demands~\cite{harris2019hierarchical}.

\section*{Discussion}

We introduced \textit{Causal Inference in Time Series} (CITS), a novel algorithm for inferring statistically causal relationships in stationary multivariate time series governed by structural causal models of arbitrary finite Markov order. By identifying the direct causes of a variable at time $t$ through conditional independence tests within a $2\tau$ time window, CITS overcomes key limitations of classical approaches such as Granger causality and PC-based methods. Unlike these earlier frameworks, CITS does not rely on parametric assumptions or fixed autoregressive structures and instead operates under a general, interpretable model of causal dependencies across time. We established both theoretical guarantees and empirical strengths of CITS, demonstrating its performance in diverse simulation settings and its application to large-scale neural data.

These findings suggest that CITS not only captures directed statistical dependencies but also reveals biologically meaningful patterns of interaction in neural circuits. Inferred causal graphs varied systematically with stimulus complexity and temporal structure, highlighting the dynamic reorganization of brain networks in response to sensory input (Fig.~\ref{fig:stimtypegraphs}). Rich stimuli such as natural scenes elicited large-scale, cross-regional connectivity linking visual, thalamic, and hippocampal regions. In contrast, simpler stimuli, such as Gabor patches, produced sparse, localized graphs. The emergence of mesoscale motifs, including strongly connected components across areas, further indicates that causal interactions are not randomly distributed but form stimulus-specific modules (Fig.~\ref{fig:stimtypegraphs}B). This aligns with theories of distributed processing, where coherent computations emerge from flexible and context-sensitive neural assemblies~\cite{Sporns2005,Bassett2017}. Our findings are consistent with prior work showing that stimulus complexity and structure shape cortical, thalamic, and hippocampal engagement~\cite{Saleem2018,Siegle2021}.

In contrast to Granger causality, which assumes linearity and can overfit noisy data, or PC-based algorithms, which are not designed for time series, CITS provides a unified, nonparametric approach tailored for temporally structured systems. The algorithm also extends beyond the recently proposed TPC framework~\cite{biswas2022statistical2} by avoiding assumptions of conditional DAGs within fixed windows and enabling inference under broader Markovian structures~\cite{biswas2022consistent}. Whereas previous methods often struggle with latent confounding or nonlinear dependencies, CITS remains robust and accurate under both conditions (see Theorem \ref{thm:main} and Corollary \ref{cor:latent}). %Moreover, the correlation between CITS-inferred edges and tuning similarity, as captured by DSI and OSI, reinforces its capacity to uncover functionally coherent circuits rather than merely statistical associations (see Fig.~\ref{fig:bootstrap}). This tuning-aligned connectivity pattern is consistent with earlier theoretical work linking information flow to feature-selective circuit motifs~\cite{Ringach2002,Shadlen1998}.

Despite these strengths, CITS has several limitations. Firstly, the mathematical guarantees of CITS have been proved under stationarity of the underlying time series and fixed causal structure, that can be applicable to short recordings at high resolution. Many biological systems exhibit non-stationary dynamics, including learning, adaptation, or state changes. Applying CITS over sliding or trial-based windows yields an evolving CFC over time, but formal extensions to model time-varying causal graphs remain a compelling direction for future work. Second, the method relies on the accuracy and power of the conditional independence tests used. While we established that Pearson's partial correlation and the Hilbert-Schmidt Independence Criterion based conditional dependence test yields accurate inference asymptotically over time series samples, it remains to establish the sensitivity of these tests in the high dimension low sample size regime. Third, the Markov order $\tau$ must be specified in advance or estimated based on estimation of order of stationary time series. Future extensions of CITS could incorporate stimulus regressors or leverage decorrelated stimulus designs to separate extrinsic from intrinsic sources of coupling.

It is also important to emphasize that statistical causality does not necessarily imply direct anatomical synaptic connectivity. Inferred links may reflect shared upstream drivers, or coordinated dynamics across populations rather than monosynaptic projections. Integrating CITS with anatomical or interventional data, such as optogenetic stimulation, lesion studies, or tract-tracing, would provide additional validation and refinement. Furthermore, while we demonstrate CITS in the context of neural recordings, the method is broadly applicable to any multivariate time series, including domains such as economics, climate systems, and behavioral ecology.

CITS also offers a framework for studying dynamic brain-wide interactions across a variety of cognitive and clinical contexts. For example, applying CITS to recordings from models of chronic pain, depression, or neuropsychiatric disorders could help uncover disruptions in top-down control circuits or altered causal flow across prefrontal, limbic, and sensory systems. This approach may offer a principled way to dissect circuit-level mechanisms underlying affective or cognitive dysfunction, potentially revealing new intervention targets grounded in statistically causal network architecture. Together, these directions position CITS as a foundational tool for discovering interpretable and robust causal structure in complex temporal systems.

\section{Methods}\label{sec:methods}
\subsection{Description of the Visual Coding Neuropixels Dataset}\label{sec:datadescriptionfull}

In this section, we provide more details on the neuropixels dataset considered in Section IV.B of the main paper. 

\subsubsection{The Dataset} We restrict our analysis to a 116 days old male mouse (Session ID 791319847), having 555 neurons recorded simultaneously by six Neuropixel probes. The spike trains for this experiment were recorded at a frequency of 1 KHz. The spike trains of this mouse are then studied across three types of stimuli: natural scenes, static gratings and Gabor patches. 

The set of stimuli ranges from natural scenes that evoke a mouse's natural habitat, to artificial stimuli such as static gratings and Gabor patches. Static gratings consist of sinusoidal patches, and Gabor patches consist of sinusoidal patches with decreasing luminosity as the distance from the center increases. By using these three stimuli, we aim to investigate how they elicit different patterns of neuronal connectivity. Dynamic stimuli such as natural movies and drifting gratings are excluded from our analysis, since their results need a more nuanced study, which we plan to conduct in the future.

\subsubsection{Description of the Stimuli}%\label{datadescription}
 In this section, we give a detailed description of each of the three stimuli considered in the neuropixels dataset.
\begin{enumerate}
    \item \textit{Natural scenes}, consisting of 118 natural scenes selected from three databases (Berkeley Segmentation Dataset, van Hateren Natural Image Dataset, and McGill Calibrated Colour Image Database) as one of the stimuli. Each scene is displayed for 250ms, after which it is replaced by the next scene in the set. The scenes are repeated 50 times in a random order with blank intervals in between.
    \item \textit{Static gratings} are full-field sinusoidal gratings with 6 orientations, 5 spatial frequencies, and 4 phases, resulting in 120 stimulus conditions. Each grating lasts for 250ms and then replaced with a different orientation, spatial frequency, and phase condition. Each grating is repeated 50 times in random order with blank intervals in between.
    \item \textit{Gabor patches}, for which the patch center is located at one of the points in a 9 × 9 visual field and three orientations are used. Each patch is displayed for 250ms, followed by a blank interval, and then replaced with a different patch. Each patch is repeated 50 times in random order with intermittent blank intervals.
\end{enumerate}

\subsubsection{Preprocessing} We transformed the recorded spike trains from 1 KHz to a bin size of 10 ms by grouping them based on the start and end times of each stimulus presentation. This allowed us to obtain the Peri-Stimulus Time Histograms (PSTHs) with a bin size of 10 ms. We then used a Gaussian smoothing kernel with a bandwidth of 16ms to smooth the PSTHs for each neuron and each stimulus presentation. The smoothed PSTHs were used as input for inferring the functional connectivity (FC) between neurons for each stimulus presentation. To select the active neurons for each stimulus, we first chose the set of neurons that were active in at least 25\% of the bins in the PSTH, and then collected the unique set of neurons over all stimuli. We found that there were 54, 43, 19, and 36 active neurons for natural scenes, static gratings, and Gabor patches, respectively, and a total of 68 unique active neurons overall. We separated the entire duration of stimulus presentation to obtain 58 trials of natural scenes, 60 trials of static gratings, 58 trials of Gabor patches, where each trial lasted for 7.5 s.

\subsection{Evaluation Metrics}\label{sec:methods_metrics}

Let True Positive (TP) represent the number of correctly identified edges, True Negative (TN) represent the number of correctly identified missing edges, False Positive (FP) represent the number of incorrectly identified edges, and False Negative (FN) represent the number of incorrectly identified missing edges across simulations. 

The \emph{Inverse False Positive Rate (IFPR)} is defined as:
\[
\text{IFPR} = \left(1 - \frac{\text{FP}}{\text{FP} + \text{TN}}\right) \cdot 100,
\]
which measures the proportion of correctly identified missing edges out of all true missing edges. Note that IFPR is scaled such that $100\%$ corresponds to perfect recovery of missing edges (i.e., no false positives).

The \emph{True Positive Rate (TPR)} is defined as:
\[
\text{TPR} = \left(\frac{\text{TP}}{\text{TP} + \text{FN}}\right) \cdot 100,
\]
which measures the proportion of correctly identified edges out of all true edges.

The **Combined Score (CS)**, also referred to as **Youden’s Index** \cite{vsimundic2009measures}, is given by:
\[
\text{CS} = \text{TPR} - \text{FPR},
\]
where $\text{FPR} = \frac{\text{FP}}{\text{FP} + \text{TN}}$.

\subsubsection{Conditional Dependence Tests}

In the Gaussian setting, we use partial correlation-based conditional dependence tests. These tests use a fixed significance level $\alpha$ and are based on the statistic:
\[
\sqrt{n - k - 3} \cdot \log\left(\frac{1 + \hat{\rho}}{1 - \hat{\rho}}\right) \leq \Phi^{-1}(1 - \alpha),
\]
where $\hat{\rho}$ is the estimated partial correlation, $n$ is the number of samples, and $k$ is the conditioning set size. This formulation is consistent with implementations in software such as the \texttt{pcalg} package in \texttt{R} and \texttt{TETRAD$^{IV}$}. The threshold can be equivalently expressed as:
\[
\gamma = \frac{\Phi^{-1}(1 - \alpha)}{\sqrt{n - k - 3}}.
\]

In the non-Gaussian setting, we use the Hilbert-Schmidt Independence Criterion for conditional dependence testing.

\section*{Code Availability}
The software package for the methods in this paper and example code are available at \url{https://github.com/abbasilab/cits}.

\section*{Data Availability}
The data used for the analyses in this paper are available at \url{https://portal.brain-map.org/circuits-behavior/visual-coding-neuropixels}.

\section*{Competing Interests}
The authors declare no competing interests.

\section*{Acknowledgments}
We thank members of the Abbasi Lab for their insightful feedback on this work. This study was partially supported by the National Institute of Mental Health of the National Institutes of Health under award number RF1MH128672 (original model development and initial validations), the National Library of Medicine of the National Institutes of Health under award number R01LM014619 (neuronal motif analysis), and the Sandler Program for Breakthrough Biomedical Research, which is partially funded by the Sandler Foundation. The content is solely the responsibility of the authors and does not necessarily represent the official views of the National Institutes of Health.

\bibliographystyle{unsrt}
\bibliography{main}

\newpage
\section*{Supplementary Material}
\addcontentsline{toc}{section}{Supplementary Material}
\setcounter{section}{0}
\renewcommand{\thesection}{S\arabic{section}}
\setcounter{figure}{0}
\renewcommand{\thefigure}{S\arabic{figure}}
\setcounter{table}{0}
\renewcommand{\thetable}{S\arabic{table}}

\section{Causal Inference in Time Series - A Review}\label{appsec:cinf_tsreview}
Among methodologies for causal inference in the time series scenario, one of the foremost is Granger Causality \cite{granger2001essays}. Granger Causality became popular as a parametric model-based approach that uses a vector autoregressive model for the time series data whose non-zero coefficients indicate causal effect between variables. In recent years, there have been non-linear extensions to Granger Causality \cite{tank2021neural}. Transfer Entropy is a non-parametric approach equivalent to Granger Causality for Gaussian processes \cite{barnett2009granger}.

A different framework for causal inference is the well-known Directed Graphical Modeling framework, which models causal relationships between variables by the Directed Markov Property with respect to a Directed Acyclic Graph (DAG) \cite{pearl2009causality}. It is a popular framework for independent and identically distributed (i.i.d.) datasets. Inference in this framework is either constraint-based or score-based. Constraint-based methods are based on conditional independence (CI) tests such as the SGS algorithm \cite{verma2022equivalence}, its faster incarnation of PC algorithm \cite{kalisch2007estimating}, both assuming no latent confounders, and the FCI algorithm in the presence of latent confounders \cite{spirtes1999algorithm}. On the other hand, score-based methods perform a search on the space of all DAGs to maximize a goodness-of-fit score, for example, the Greedy Equivalence Search (GES) and Greedy Interventional Equivalence Search (GIES) \cite{hauser2012characterization}. However, these approaches are based on i.i.d. observations, and their extension to time series settings is not trivial. In fact, naive application of PC algorithm to time series data is seen to suffer in performance due to not incorporating across-time causal relations \cite{biswas2022statistical1}.

Recently the PC algorithm has been extended to form the TPC algorithm that is applicable to time series datasets \cite{biswas2022statistical2}. Entner and Hoyer suggest to use the FCI algorithm for inferring the causal relationships from time series in the presence of unobserved variables \cite{entner2010causal}. The PCMCI algorithm \cite{runge2019detecting}, another variation of the PC algorithm,  consists of the following two steps. In the first step, it uses the PC algorithm to detect parents of a variable at a time $t$, and in the second step, it applies a momentary conditional independence (MCI) test. Chu et al. \cite{chu2008search} propose a causal inference algorithm based on conditional independence, designed for additive, nonlinear and stationary time series, using two kinds of conditional dependence tests based on additive regression model. Using the fact that one can incrementally construct the mutual information between a cause and an effect, based on the mutual information values between the effect and the previously found causes, Jangyodsuk et al. \cite{jangyodsuk2014causal} proposes to obtain a causal graph with each time series as a node and the edge weight for each edge is the difference in time steps between the occurrence of cause and the occurrence of the effect. Amortized Causal Inference is another conditional independence-based algorithm, that infers causal structure by performing amortized variational inference over an arbitrary data-generating distribution \cite{lorch2022amortized}.

%\section{Proofs of the Main Results}
\section{Theoretical Guarantees}\label{appsec:theo_guar}
In this section, we establish consistency of the CITS algorithm in recovering the true causal relationships from time series data under mild conditions on the underlying time series.

Let $\mu(X_{u,s},X_{v,t}\vert\bm{X}_S)$ denote a measure of conditional dependence of $X_{u,s}$ and $X_{v,t}$ given $\bm{X}_S$, i.e. it takes value $0$ if and only if we have conditional independence, and let $\hat{\mu}_n(X_{u,s},X_{v,t}\vert\bm{X}_S)$ be its consistent estimator. In Theorem \ref{thm:main}, we will use $\hat{\mu}_n(X_{u,s},X_{v,t}\vert\bm{X}_S)$ to construct a conditional dependence test guaranteeing consistency of the CITS estimate.

\setcounter{theorem}{2} % or whatever number Theorem 1 was

\renewcommand{\thetheorem}{2}

\begin{theorem}[CITS Consistency for Time-Unrolled and Rolled Graphs (restated)]
Let $\{\bm{X}_t\}_{t\in \mathbb{Z}}$ be a strictly stationary stochastic process of finite Markov order $\tau$ that follows a time-invariant structural causal model with DAG $G$ as in \eqref{eqdef: process}, and assume the distribution is faithful to $G$. Let $\hat{G}_n$ and $\hat{G}_{n,R}$ denote the estimated unrolled and rolled graphs, respectively, obtained by the sample CITS algorithm using a consistent conditional dependence test:
\[
\left\vert\hat{\mu}_n(X_{u,s},X_{v,t} \mid \bm{X}_S)\right\vert > \gamma
\]
for some fixed $\gamma > 0$, where $\hat{\mu}_n$ consistently estimates a valid conditional dependence measure.

Then, as $n \rightarrow \infty$:
\begin{itemize}
    \item $\hat{G}_n$ recovers the true time-unrolled DAG $G$ up to the Markov equivalence class over concurrent edges (i.e., for edges with $s = t$), and exactly for all non-concurrent edges (i.e., $s < t$).
    \item $\hat{G}_{n,R}$ recovers the rolled graph $G_R$ with all non-concurrent edges recovered exactly, and concurrent edges recovered up to their Markov equivalence class.
\end{itemize}
\end{theorem}

In this section, we are going to consider two consistent estimators of the conditional dependence measures, namely the (log-transformed version of the) partial correlation $Z_n$ in the Gaussian setting, and the Hilbert-Schmidt estimator $H_n$ (see Sections 2.3 and 2.4  in \cite{biswas2022consistent}) in the non-Gaussian setting. For a detailed account of these conditional dependence estimators, see Supplementary Section \ref{appsec:condindtests}.

In order to establish consistency of the estimators $Z_n$ and $H_n$, under the non-i.i.d. setting, one needs to assume some mixing conditions on the underlying time series. Two such standard mixing conditions are $\rho$-mixing and $\alpha$-mixing \cite{kolmogorov1960strong,bradley2005basic} (see Supplementary Section \ref{appsec:defmixing} for the definitions of these two notions of mixing).  We make either of the following two assumptions related to mixing of the time series to ensure consistency of the estimators.

\subsection*{Assumption 1}\label{rhoassumptions}
\begin{enumerate}
    \item $\{X_{u,t},X_{v,t}: t=1,2,\ldots\}$ is $\rho$-mixing for all $u,v$, with maximal correlation coefficients $\xi_{uv}(k), k\ge 1$.
     \item $\e X_{u,t}^4 <\infty$ for all $u,t$ and $\sum_{k=1}^{\infty} \xi_{uv} (k) < \infty$ for $u,v\in 1,\dots,p$.
    \item There exists a sequence of positive integers $s_n\rightarrow \infty$ and $s_n = o(n^{1/2})$ such that $n^{1/2}\xi_{uv}(s_n)\rightarrow 0$ as $n\rightarrow\infty$ for $u,v\in 1,\ldots,p$.
\end{enumerate}

\subsection*{Assumption 1*}\label{alphaassumptions}
\begin{enumerate}
    \item[$1^*$.] $\{X_{u,t},X_{v,t}: t=1,2,\ldots\}$ is strongly mixing with coefficients $\{\alpha_{uv}(k)\}_{k\geq 1}$ for all $u,v=1,\ldots,p$.
    \item[$2^*$.] $\e |X_{v,t}|^{2\delta} < \infty$ for some $\delta > 2$ and all $v,t$, and the strongly mixing coefficients satisfy: $\sum_{k=1}^{\infty} \alpha_{uv} (k)^{1-2/\delta} < \infty$ for $u,v= 1,\dots,p$.
    \item[$3^*$.] There exists a sequence of positive integers $s_n\rightarrow \infty$ and $s_n = o(n^{1/2})$ such that $n^{1/2}\alpha_{uv}(s_n)\rightarrow 0$ as $n\rightarrow\infty$ for $u,v= 1,\ldots,p$.
\end{enumerate}
Assumptions 1.1-1.3 and $1^*.1^*-1^*.3^*$ are adapted from Conditions 1-2 in \cite{masry2011estimation}. See Remark 4.1 in \cite{biswas2022consistent} for a comparative discussion between these two sets of assumptions.

\begin{lemma}\label{lemma:consistency}
Under either Assumption \hyperref[rhoassumptions]{1} or Assumption \hyperref[alphaassumptions]{$1^*$}, 
\begin{enumerate}
    \item $Z_n(X_{u,s},X_{v,2\tau+1}\vert X_{S})$ converges to $z(X_{u,s},X_{v,2\tau+1}\vert X_{S})$ in probability.
    
    \item If the regularization constant $\epsilon_n$ satisfies $n^{-1/3} \ll \epsilon_n \ll 1$, then, $H_n(X_{u,s},X_{v,2\tau+1}\vert X_{S})$ converges to $H(X_{u,s},X_{v,2\tau+1}\vert X_{S})$ in probability,
\end{enumerate}
where $u,v=1,\ldots,p$, $s=\tau+1,\ldots,2\tau$ and $S\subset\{(d,r):d=1,\ldots,p;r=1,\ldots,2\tau+1\}\setminus\{(u,s),(v,2\tau+1)\}$.
\end{lemma}
\begin{proof}
The proof directly follows from Lemmas 3.3 and 3.4 in \cite{biswas2022consistent}. 
\end{proof}
\begin{corollary}
Under Assumption \hyperref[rhoassumptions]{1} or Assumption \hyperref[alphaassumptions]{$1^*$}, using sample CITS based on either of these two conditional dependence tests: 1) partial correlation for Gaussian regime and 2) Hilbert-Schimdt conditional dependence criterion for non-Gaussian regime, leads to asymptotically accurate estimation of the DAG $G$ and Rolled DAG $G_R$ (see Theorem \ref{thm:main}). This is because, the sample z-transformed partial correlation for the Gaussian regime and the Hilbert-Schimdt conditional dependence criterion for the non-Gaussian regime, form consistent estimators to their corresponding conditional dependence measures (see Lemma \ref{lemma:consistency}).
\end{corollary}

\subsection{Proof of Lemma \ref{lemma:concept}}\label{proofLem3.1}
To see that (1) implies (2), let (2) be false, i.e. $X_{v,t}$ and $X_{u,s}$ are adjacent in $G$. Then by the time order it follows that $X_{u,s}\rightarrow X_{v,t}$, which implies (1) is false. 
    
(2) implies (1) holds trivially. 
    
Next, suppose that (2) holds, i.e. $X_{v,t}$ and $X_{u,s}$ are non-adjacent in $G$. Then $X_{v,t}$ and $X_{u,s}$ are d-separated by the set of their parents in $G$. Next, note that the parents of $X_{v,t}$ and $X_{u,s}$ are between times $\{t-\tau,\ldots,t-1\}$ and $\{s-\tau,\ldots,s-1\}$ respectively and hence, the parents of both $X_{v,t}$ and $X_{u,s}$ are within times $\{t-2\tau,\ldots,t-1\}$. This implies that $X_{v,t}$ and $X_{u,s}$ are d-separated by $S_0:= \text{pa}(v,t)\cup \text{pa}(u,s) \subseteq \{(d,r):d=1,\ldots,p; r= t-2\tau,\ldots,t-1\}$. Since the SCM implies directed Markov property \cite{lauritzen1996graphical}, so it follows that $X_{v,t}\ind X_{u,s} \vert \bm{X}_{S_0}$. We thus showed that (2) implies (3).

Finally, let (3) hold. Then under faithfulness, it follows that $X_{v,t}$ and $X_{u,s}$ are d-separated in $G$ by $S$ which implies (2). 

\subsection{Proof of Theorem \ref{thm:CITS_oracle}}\label{proofThmOracle}
    Since the DAG is time invariant, in order to obtain $G$, it suffices to obtain parental sets of variables $\text{pa}(v,t)$ at a fixed time $t=2\tau+1$ and $v=1,\ldots,p$. This justifies to inputting times $t=1,\ldots,2\tau+1$ to the algorithm \ref{alg:CITS_oracle} to obtain $G$.

    Line 1 of algorithm \ref{alg:CITS_oracle} initializes the DAG. For any $u,v=1,\ldots,p$, if there is an edge in $G$ from $X_{u,s}$ to $X_{v,t}$, then $s\leq t$ holds, due to time order. This justifies initializing with only edges $X_{u,s}\rightarrow X_{v,t}$ with $s\leq t$.
    
    Lines 2-9 consider each possible edge $X_{u,s}\rightarrow X_{v,2\tau+1}$ and deletes the edge if for some $S\subseteq\{(d,r):d=1,\ldots,p; r=1,\ldots,2\tau+1\}$, $X_{u,s}\ind X_{v,2\tau+1}\vert \bm{X}_S$, where $\bm{X}_{S} = \{X_{d,r}: (d,r)\in S\}$. This is justified by the implication (3)$\implies$ (1) in Lemma \ref{lemma:concept} concluding that $X_{u,s}\not\in \text{pa}(v,2\tau+1)$.
    
    In fact, by Lemma \ref{lemma:concept} implication (1) $\implies$ (3), it follows that for the remaining edges $X_{u,s}\rightarrow X_{v,2\tau+1}$, $X_{u,s}$ will be the parents of $X_{v,2\tau+1}$, $u=1,\ldots, p$, $s=\tau+1,\ldots,2\tau+1$. This justifies that line 10 correctly finds the parental set $\text{pa}(v,2\tau+1)$ from the remaining edges in $G$.

    Based on the parental sets $\text{pa}(v,2\tau+1)$, $v=1,\ldots,p$, line 11 directs the edges from $\text{pa}(v,2\tau+1)$ to $X_{v,2\tau+1}$ to outputs the DAG $G$ and line 12 converts $G$ to its Rolled Graph $G_R$. If \( s < 2\tau+1 \), time order from the initialization in Line 1 uniquely identifies the direction. However, when \( s = 2\tau+1 \), directionality is not resolved by the initialization; however since the SCM is assumed to satisfy faithfulness, conditional independence information would identify the concurrent edges up to their Markov equivalence class.

\subsection{Proof of Theorem \ref{thm:main}}\label{proofTh4.1}
For a pair of nodes $u,v=1,\ldots,p$ and times $s=\tau+1,\ldots,2\tau+1$ and a conditioning set $S\subseteq \mathbb{S}:=\{(d,r):d=1,\ldots,p;r=1,\ldots,2\tau+1\}\setminus\{(u,s),(v,2\tau+1)\}$, let $E_{u,v,s\vert S}$ denote an error event that occurred when testing conditional dependence of $X_{v,2\tau+1}\ind X_{u,s}\vert \bm{X}_S$, i.e.,

\[
E_{u,v,s\vert S} = E_{u,v,s\vert S}^{I} \cup  E_{u,v,s\vert S}^{II},
\]

where
\begin{multline*}
E_{u,v,s\vert S}^{I} := \{|\hat{\mu}_n(X_{u,s},X_{v,2\tau+1}\vert\bm{X}_S)| > \gamma \\\text{ and } \mu(X_{u,s},X_{v,2\tau+1}\vert\bm{X}_S) = 0\}    
\end{multline*}
\begin{multline*}
E_{u,v,s\vert S}^{II} := \{ |\hat{\mu}_n(X_{u,s},X_{v,2\tau+1}\vert\bm{X}_S)| \le \gamma \\\text{ and } \mu(X_{u,s},X_{v,2\tau+1}\vert\bm{X}_S) \neq 0\}    
\end{multline*}
denote the events of Type I error and Type II error, respectively. Thus,

\begin{align*}
&P(\text{an error occurs in the sample CITS algorithm})\\
&\leq P\left(\bigcup_{u,v,s,S\subseteq \mathbb{S}}E_{u,v,s\vert S}\right)\\
& \leq O(1)\sup_{u,v,s,S\subseteq\mathbb{S}}P(E_{u,v,s\vert S}) \numberthis\label{eqn:skelerror}   
\end{align*}
using that the cardinality of the set $\vert \{u,v,s,S\subseteq \mathbb{S}\} \vert = p^2\tau 2^{2p\tau-2}$. Then, for any $\gamma > 0$, we have:

\begin{align}\label{t1err}
    \sup_{u,v,s,S\subseteq\mathbb{S}} P(E_{u,v,s\vert S}^{I}) \le \sup_{u,v,s,S\subseteq \mathbb{S}} P(\vert \hat{\mu}_n(X_{u,s},X_{v,2\tau+1}\vert\bm{X}_S)\notag\\ - \mu(X_{u,s},X_{v,2\tau+1}\vert\bm{X}_S)\vert > \gamma) = o(1)
\end{align}
by the consistency of $\hat{\mu}_n$.

Next, we bound the type II error probability. Towards this,  let $c=\inf \{\vert\mu(X_{u,s},X_{v,t}\vert\bm{X}_S)\vert : \mu(X_{u,s},X_{v,t}\vert\bm{X}_S)\neq 0, u,v=1,\ldots,p;s=\tau+1,\ldots,2\tau+1;S\subseteq \{(d,r):d=1,\ldots,p;r=1,\ldots,2\tau+1\}\setminus\{(u,s),(v,2\tau+1)\}\}>0$, and choose $\gamma = c/2$. Then,

\begin{multline}
    \sup_{u,v,s,S\subseteq\mathbb{S}} P(E_{u,v,s\vert S}^{II})\\ =\sup_{u,v,s,S\subset \mathbb{S}} P(\vert \hat{\mu}_n(X_{u,s},X_{v,2\tau+1}|S)\vert \leq \gamma ~,\\ \hspace{1in}\mu(X_{u,s},X_{v,2\tau+1}|S)\neq 0)\\
\leq \sup_{u,v,s,S\subset \mathbb{S}} P(\vert \hat{\mu}_n(X_{u,s},X_{v,2\tau+1}|S)\\
\hspace{1in}- \mu(X_{u,s},X_{v,2\tau+1}|S) \vert \ge c/2))\\
=o(1) \hspace{2in}\label{eqn:skeltype2err}.
\end{multline}

It follows from \eqref{eqn:skelerror}, \eqref{t1err} and \eqref{eqn:skeltype2err}, that:
\begin{equation}\label{eq: main_last}
P(\text{an error occurs in the sample CITS algorithm})\rightarrow 0
\end{equation}
The event of no error ocurring in the sample CITS algorithm is same as the event that the outcome of sample CITS would be same as the CITS-Oracle which has knowledge of conditional independence information. By Theorem \ref{thm:CITS_oracle}, the CITS-Oracle outputs the true DAG $G$ and its Rolled DAG $G_R$. In summary, the event of no error occurring in the sample CITS algorithm implies is equivalent to stating that $\hat{G}=G$ and $\hat{G}_R=G_R$. The proof of Theorem \ref{thm:main} is now complete, in view of \eqref{eq: main_last}.

% \subsection{Proof of Corollary \ref{corrolary:robustlatent}}

\section{Choice of Conditional Dependence Tests}\label{appsec:condindtests}
In this section, we describe some choices of the conditional dependence tests used in the CITS algorithm. According to Theorem \ref{thm:main}, we can use a conditional dependence test of the following form \eqref{eq:cond_test} in the CITS algorithm, 
\begin{equation}\label{eq:cond_test}
\vert\hat{\mu}_n(X_{u,s},X_{v,t}\vert\bm{X}_S)\vert>\gamma
\end{equation} to guarantee its consistency, as long as it satisfies the condition: $\hat{\mu}_n(X_{u,s},X_{v,t}\vert\bm{X}_S)$ is a consistent estimator of $\mu(X_{u,s},X_{v,t}\vert\bm{X}_S)$, the latter being a measure of conditional dependence of $X_{u,s}$ and $X_{v,t}$ given $\bm{X}_S$. In the following sections, we provide examples of such candidates for $\hat{\mu}_n(X_{u,s},X_{v,t}\vert\bm{X}_S)$ and resulting conditional dependence tests, in both the Gaussian as well as the non-Gaussian regime.
\subsection{The Gaussian Regime: Pearson's Partial Correlations}\label{sec:samplepc} 
It is popular to use partial correlations to test conditional dependence for causal inference, such as in the PC algorithm \cite{kalisch2007estimating}. The partial correlation-based conditional dependence test is applicable in the Gaussian setting. Assume that $\mathbf{Y}=(Y_1,\ldots,Y_p)$ is a $p$-dimensional Gaussian random vector, for some fixed integer $p$. For $i \neq j \in \{1,\ldots ,p\},\ \bm{k} \subseteq \{1,\ldots ,p\}
  \setminus 
\{i,j\}$, denote by $\rho(Y_i,Y_j|Y_{\bm{k}})$ the 
partial correlation between $Y_i$ and $Y_j$ given
$\{Y_r:\ r \in \bm{k}\}$. The partial correlation serves as a measure of conditional dependence in the Gaussian regime in view of the following standard property of the multivariate Gaussian distribution (see Prop. 5.2 in \cite{lauritzen1996graphical}),

\[\rho(Y_i,Y_j|Y_{\bm{k}})=0\text{ if and only if }Y_i\ind Y_j ~\vert~ \{Y_r~:\ r \in \bm{k}\}.\] 

\medskip
Denote $k=\vert \bm k\vert$ and let without loss of generality $\{Y_r;\ r \in \bm{k}\}$ be the last $k$ entries in $\bm{Y}$. Let $\Sigma:= \text{cov}(\bm{Y})$ with $\Sigma = \left(\begin{array}{cc}
    \Sigma_{11} & \Sigma_{12} \\
    \Sigma_{21} & \Sigma_{22}
\end{array}\right)$ where $\Sigma_{11}$ is of dimension $(p-k)\times (p-k)$, $\Sigma_{22}$ is of dimension $k \times k$, and $\Sigma_{11.2}=\Sigma_{11}-\Sigma_{12}\Sigma_{22}^{-1}\Sigma_{21}$. Let $\bm e_1,\ldots,\bm e_p$ be the canonical basis vectors of $\mathbb{R}^p$. It follows from \cite{muirhead2009aspects} that,
\begin{align*}
\rho(Y_i,Y_j|Y_{\bm{k}}) &= \frac{\bm{e}_i^{\top}\Sigma_{11.2}\bm{e}_j}{\sqrt{(\bm{e}_i^{\top}\Sigma_{11.2}\bm{e}_i(\bm{e}_j^{\top}\Sigma_{11.2}\bm{e}_j))}}
\end{align*}
One can calculate the sample partial correlation
$\hat{\rho}(Y_i,Y_j|Y_{\bm{k}})$ via regression or by using the following identity, with  $\hat{\Sigma}$ and $\hat{\Sigma}_{11.2}$ being the sample versions of $\Sigma$ and $\Sigma_{11.2}$.
\begin{align*}
\hat{\rho}(Y_i,Y_j|Y_{\bm{k}}) &= \frac{\bm{e}_i^{\top}\hat{\Sigma}_{11.2}\bm{e}_j}{\sqrt{(\bm{e}_i^{\top}\hat{\Sigma}_{11.2}\bm{e}_i)(\bm{e}_j^{\top}\hat{\Sigma_{11.2}}\bm{e}_j)}}
\end{align*}
For testing whether a partial correlation is zero or not, we first apply Fisher's z-transform
\begin{eqnarray}\label{ztrans}
Z_n(Y_i,Y_j|Y_{\bm{k}}):= g(\hat{\rho}(Y_i,Y_j|Y_{\bm{k}})) :=\notag\\\frac{1}{2} \log \left (\frac{1 +
    \hat{\rho}(Y_i,Y_j|Y_{\bm{k}})}{1 - \hat{\rho}(Y_i,Y_j|Y_{\bm{k}})} \right).
\end{eqnarray}

Also, let $z(Y_i,Y_j|Y_{\bm{k}}) = g(\rho(Y_i,Y_j|Y_{\bm{k}}))$. Note that $z(Y_i,Y_j|Y_{\bm{k}})=0 \Leftrightarrow \rho(Y_i,Y_j|Y_{\bm{k}})=0$, and hence, $z(Y_i,Y_j|Y_{\bm{k}}) = 0 \Leftrightarrow Y_i\ind Y_j ~\vert~ \{Y_r~:\ r \in \bm{k}\}$. That is, $z(Y_i,Y_j|Y_{\bm{k}})$ is also a measure of conditional dependence. Furthermore, $Z_n(Y_i,Y_j|Y_{\bm{k}})$ is a consistent estimator of the conditional dependence measure $z(Y_i,Y_j|Y_{\bm{k}})$ (See Lemma III.1). Hence, it can be used to construct a conditional dependence test of the form $Z_n(Y_i,Y_j|Y_{\bm{k}}) > \gamma$ for some fixed $\gamma>0$, for the CITS in the Gaussian regime.

\subsection{The Non-Gaussian Regime: The Hilbert Schmidt Criterion.}\label{sec:samplepch}

In the general non-Gaussian scenario, zero partial correlations do not necessarily imply conditional independence. In such cases, the Hilbert Schmidt criterion can be used as a convenient test for conditional dependence, which is described in more details below.

Given $\mathbb{R}$-valued random variables $X,Y$ and the random vector $\bm Z$, we propose to use the following statistic for testing the conditional dependence of $X,Y\vert \bm{Z}$ (see \cite{fukumizu2007kernel}):

\[
H_{n}(X,Y\vert\bm{Z}) = Tr[R_{\overset{..}{Y}}R_{\overset{..}{X}}-2R_{\overset{..}{Y}}R_{\overset{..}{X}}R_{\bm Z} + R_{\overset{..}{Y}}R_{\bm Z}R_{\overset{..}{X}}R_{\bm Z}]
\]
where $R_X := G_X(G_X+n\epsilon_n I_n)^{-1}, R_Y := G_Y(G_Y+n\epsilon_n I_n)^{-1}, R_{\bm Z} := G_{\bm Z}(G_{\bm Z}+n\epsilon_n I_n)^{-1}$, $G_X,G_Y,G_{\bm Z}$ being the centered Gram matrices with respect to a positive definite and integrable kernel $k$, i.e. $G_{X,ij}=<k(\cdot,X_i)-\hat{m}_X^{(n)},k(\cdot,X_j)-\hat{m}_X^{(n)}>$ with $\hat{m}_X^{(n)}=\frac{1}{n}\sum_{i=1}^n k(\cdot,X_i)$, and $\overset{..}{X} := (X,{\bm Z}), \overset{..}{Y} := (Y,{\bm Z})$. Under some regularity assumptions mentioned below, it follows from the proof of Theorem 5 in \cite{fukumizu2007kernel} that $\hat{H}_n(X,Y\vert \bm{Z})$ is a consistent estimator of $H(X,Y\vert \bm Z):= \|V_{\overset{..}{X}\overset{..}{Y}|{\bm Z}}\|^2$, where $$V_{\overset{..}{X}\overset{..}{Y}|{\bm Z}} := \Sigma_{\overset{..}{X}\overset{..}{X}}^{-1/2}(\Sigma_{\overset{..}{X}\overset{..}{Y}} - \Sigma_{\overset{..}{X}\bm Z} \Sigma_{\bm Z \bm Z}^{-1}\Sigma_{\bm Z \overset{..}{Y}}) \Sigma_{\overset{..}{Y} \overset{..}{Y}}^{-1/2}$$ and $\Sigma_{UV}$ denotes the covariance matrix of $U$ and $V$. It follows from \cite{fukumizu2007kernel} that $X\ind Y |\bm Z \Leftrightarrow H(X,Y|\bm Z)=0$, i.e. $H(X,Y|\bm Z)$ is a measure of conditional dependence of $X$ and $Y$ given $Z$. Also, it follows from Lemma III.1 that $H_n(X,Y| \bm Z)$ is a consistent estimator to $H(X,Y|\bm Z)$, and hence, $H_n(X,Y| \bm Z)$ can be used to construct a conditional dependence test in the CITS of the form: $\vert H_n(X,Y| \bm Z)\vert > \gamma$ for a fixed $\gamma>0$. This method has the advantage that unlike the Pearson partial correlation, it does not require Gaussianity of the data to decide conditional independence, and consequently can be used in CITS if the underlying time series is non-Gaussian.

\section{Two Notions of Mixing}\label{appsec:defmixing}
For fixed $u,v \in 1,\ldots, p$, let $\mathcal{F}_{a}^{b}$ be the $\sigma$-field of events generated by the random variables $\{X_{u,t},X_{v,t}:a\leq t \leq b\}$, and $L_2(\mathcal{F}_a^b)$ be the collection of all second-order random variables which are $\mathcal{F}_a^b$-measurable.

\begin{definition}[$\rho$-mixing]
In this section, we describe two common notions of mixing, that we are going to assume on our underlying time series, in order to guarantee consistency of the conditional dependence estimators. The stationary process $\{X_{u,t},X_{v,t}:t= 1,2,\ldots\}$ is called $\rho$-mixing if the maximal correlation coefficient

\begin{equation}\label{eq:rhomixing}
\xi_{uv}(k):=\sup_{\ell\geq 1}\sup_{\substack{U\in L_2(\mathcal{F}_{1}^\ell)\\ V\in L_2(\mathcal{F}_{\ell+k}^{\infty})}} \frac{\vert \text{cov}(U,V) \vert}{\text{var}^{1/2}(U)\text{var}^{1/2}(V)} \rightarrow 0 \text{ as } k\rightarrow \infty.
\end{equation}
\end{definition}

\begin{definition}[$\alpha$-mixing]
The stationary process $\{X_{u,t},X_{v,t}:t= 1,2,\ldots\}$ is called $\alpha$-mixing if:

\[
\alpha_{uv}(k):=\sup_{\ell\geq 1}\sup_{\substack{A\in L_2(\mathcal{F}_{1}^\ell)\\ B\in L_2(\mathcal{F}_{\ell+k}^{\infty})}} |P(A\cap B) - P(A) P(B)| \rightarrow 0 \text{ as } k\rightarrow \infty.
\]
\end{definition}
\section{Simulation Study Details}\label{appsec:simulstudy}
We outline the different simulation settings used in the numerical experiments.

\noindent\textbf{Linear Additive with Gaussian noise:}
\begin{itemize}
    \item Linear Gaussian Model 1: 
    \begin{multline*}
    \hspace{1in}(X_{1,t},X_{2,t},X_{3,t},X_{4,t}) = ( 1+\epsilon_{1,t},\\~-1+\epsilon_{2,t},~2X_{1,t-1}-X_{2,t-1}+\epsilon_{3,t}, ~2X_{3,t-1}+\epsilon_{4,t})
    \end{multline*}
     where $\epsilon_{i,t}, i=1,\ldots,4, t=1,2,\ldots,1000, \sim$ i.i.d. $N(0,\eta)$ with mean $0$ and standard deviation $\eta$. In this and all subsequent examples, the parameter $\eta$ is assumed to vary between $0$ to $3.5$ in increments of $0.5$ in our numerical experiments. The true Rolled graph $G_R$ for this model has the edges $1\rightarrow 3, 2\rightarrow 3, 3\rightarrow 4$.
    \item  Linear Gaussian Model 2: 
    \begin{multline*}
    (X_{1,t},X_{2,t},X_{3,t},X_{4,t}) = ( 1+\epsilon_{1,t},~-1+2X_{1,t-1}+\epsilon_{2,t},\\~2X_{1,t-1}+\epsilon_{3,t}, ~X_{2,t-1} + X_{3,t-1} +\epsilon_{4,t})
    \end{multline*}
   where $\epsilon_{i,t}, i=1,\ldots,4, t=1,2,\ldots,1000, \sim$ i.i.d. $N(0,\eta)$ with mean $0$ and standard deviation $\eta$. The true Rolled graph $G_R$ for this model has the edges $1\rightarrow 2, 1\rightarrow 3, 2\rightarrow 4, 3\rightarrow 4$.
   \end{itemize}
    
    \noindent\textbf{Non-linear Additive with Non-Gaussian noise:}
    \begin{itemize}
    \item Non-linear Non-Gaussian Model 1:     
    \begin{multline*}
    \hspace{1in}(X_{1,t},X_{2,t},X_{3,t},X_{4,t}) =  (\epsilon_{1,t},\\~\epsilon_{2,t},~4\sin(X_{1,t-1})-3\sin(X_{2,t-1})+\epsilon_{3,t},~ 3 X_{3,t-1} +\epsilon_{4,t})  
    \end{multline*}
   where $\epsilon_{i,t}, i=1,\ldots,4, t=1,2,\ldots,1000, \sim$ i.i.d. and uniformly distributed on the interval $(0,\eta)$. The true Rolled graph $G_R$ for this model has the edges $1\rightarrow 3, 2\rightarrow 3, 3\rightarrow 4$. 
    
    \item Non-linear Non-Gaussian Model 2:  
    \begin{multline*}
    (X_{1,t},X_{2,t},X_{3,t},X_{4,t}) = ( \epsilon_{1,t},~4X_{1,t-1}+\epsilon_{2,t},\\~3\sin(X_{1,t-1})+\epsilon_{3,t}, ~8\log(\vert X_{2,t-1}\vert) + 9\log(\vert X_{3,t-1}\vert) +\epsilon_{4,t})
    \end{multline*}
   where $\epsilon_{i,t}, i=1,\ldots,4, t=1,2,\ldots,1000, \sim$ i.i.d. Uniform distribution on the interval $(0,\eta)$. The true Rolled graph $G_R$ for this model has the edges $1\rightarrow 2, 1\rightarrow 3, 2\rightarrow 4, 3\rightarrow 4$.
\end{itemize}

   \noindent \textbf{Non-linear Non-additive Continuous Time Recurrent Neural Network with Gaussian noise:}
   \begin{itemize}
   \item Continuous Time Recurrent Neural Network (CTRNN) Model:  \begin{equation}\label{ctrnn}
    \tau_j \frac{dX_{j,t}}{dt}=-X_{j,t} + \sum_{i=1}^m w_{ij} \sigma (X_{i,t}) + \epsilon_{j,t}, j=1,\ldots, m,
    \end{equation}
    We consider a motif consisting of $m=4$ and $w_{13}=w_{23}=w_{34}=10$ and $w_{ij}=0$ otherwise. We also note that in Eq. \eqref{ctrnn}, $X_{j,t}$ depends on its own past. Therefore, the true Rolled graph has the edges $1\rightarrow 3,2\rightarrow 3,3\rightarrow 4, 1\rightarrow 1, 2\rightarrow 2, 3\rightarrow 3, 4\rightarrow 4$. The time constant $\tau_i$ is set to 10 units for each $i$. We consider $\epsilon_{i,t}$ to be distributed as an independent Gaussian process with mean $1$ and standard deviation $\eta$. The signals are sampled at a time gap of $g := \exp(1) \approx 2.72$ units for a total duration of $1000$ units. For simulation purposes, one may replace the continuous derivative on the left hand side of Eq. \eqref{ctrnn} is replaced by first order differences at a gap of $g$. 
\end{itemize}
\end{document}